\documentclass[10pt, two column, twoside]{IEEEtran}
\usepackage{amssymb}
\usepackage{amsmath} 
\usepackage[lined,boxed,commentsnumbered, ruled]{algorithm2e}
\usepackage{mathrsfs}
\usepackage{algorithmic}
\usepackage{bm}
\usepackage{tikz}
\usetikzlibrary{arrows}
\usepackage{subfigure}
\usepackage{graphicx,booktabs,multirow}
\usepackage{epstopdf}
\usepackage{subfigure}
\usepackage{multirow}
\usepackage{tabularx} 
\usepackage{booktabs}

\definecolor{colorhkust}{RGB}{20,43,140}
\definecolor{colortsinghua}{RGB}{116,52,129}
\definecolor{color1}{RGB}{128,0,0}

\usepackage{amsthm}

\newtheorem{lemma}{Lemma}
\newtheorem{theorem}{Theorem}

\newtheorem{proposition}{Proposition}
\newtheorem{definition}{Definition}
\newtheorem{remark}{Remark}

\newcommand{\trace}{{\rm Tr}}

\newcommand{\grad}{\mathrm{grad}}

\newcommand{\hess}{\rm{Hess}}
\newcommand{\diagg}{\mathrm{diag}}

\begin{document}
	\title{Blind Demixing for Low-Latency Communication}
\author{Jialin Dong, \textit{Student Member}, \textit{IEEE},  Kai Yang, \textit{Student Member}, \textit{IEEE}, and Yuanming~Shi, \textit{Member}, \textit{IEEE}
	
	\thanks{J. Dong, K. Yang and Y. Shi are with the School of Information Science and Technology, ShanghaiTech University, Shanghai, 201210 China (e-mail: \{dongjl, yangkai, shiym\}@shanghaitech.edu.cn).}
	
}
	
	\maketitle
\begin{abstract}
In the next generation wireless networks, low-latency communication is critical to support emerging diversified applications, e.g., Tactile Internet and Virtual Reality. In this paper, a novel blind demixing approach is developed to reduce the channel signaling overhead, thereby supporting low-latency communication. Specifically, we develop a low-rank approach to recover the original information only based on the single observed vector without any channel estimation. To address the unique challenges of multiple non-convex rank-one constraints, the quotient manifold geometry of the product of complex symmetric rank-one matrices is exploited. This is achieved by equivalently reformulating the original problem that uses complex asymmetric matrices to the one
that uses Hermitian positive semidefinite matrices. We further generalize the geometric concepts of the complex product manifold via element-wise extension of the geometric concepts of the individual manifolds. The scalable Riemannian optimization algorithms, i.e., the Riemannian gradient descent algorithm and the Riemannian trust-region algorithm, are then developed to solve the blind demixing problem efficiently with low iteration complexity and low iteration cost. Statistical analysis shows that the Riemannian gradient descent with spectral initialization is guaranteed to linearly converge to the ground truth signals provided sufficient measurements. In addition, the Riemannian trust-region algorithm is provable to converge to an approximate local minimum from arbitrary initialization point. Numerical experiments have been carried out in the settings with different types of encoding matrices to demonstrate the algorithmic advantages, performance gains and sample efficiency of the Riemannian optimization algorithms.
\end{abstract}

\section{Introduction}
Recently, various emerging 5G applications such as Internet-of-Things (IoT) \cite{al2015internet}, Tactile Internet \cite{simsek20165g} and Virtual Reality \cite{bastug2017toward} are unleashing a sense of urgency in providing \emph{low-latency} communications \cite{andrews2014will}, for which innovative new technologies need to be developed. To achieve this goal, various solutions have been investigated, which can be typically categorized into three main types, i.e., radio access network (RAN), core network , as well as mobile edge caching and computing \cite{parvez2017survey}. In particular, by pushing the computation and storage resources to the network edge, followed by network densification, dense Fog-RAN provides a principled way to reduce the latency \cite{shi2015large}. In addition, reducing the packet blocklength, e.g., short packets communication \cite{durisi2016toward}, is a promising technique in RAN to support low-latency communication, for which the theoretical analysis on the tradeoffs among the channel coding rate, blocklength and error probability was provided in \cite{polyanskiy2010channel}.

However, channel signaling overhead reduction becomes critical to design a low latency communication system. In particular, when packet blocklength is reduced as envisioned in 5G systems, channel signaling overhead dominates the major portion of the packet \cite{parvez2017survey}. Furthermore, massive channel acquisition overhead becomes the bottleneck for interference coordination in dense wireless networks \cite{shi2015large}. To address this issue, numerous research efforts have been made on channel signaling overhead reduction.
The compressed sensing based approach was developed in \cite{bajwa2010compressed}, yielding good performance with low energy, latency and bandwidth cost. The recent proposal of topological interference alignment \cite{shi2016low} serves as a promising way to manage the interference based only on the network connectivity information at the transmitters. Furthermore, by equipping a large number of antennas at the base stations, massive MIMO \cite{rusek2013scaling} can manage the interference without channel estimation at the transmitters. However, all the methods \cite{shi2016low,rusek2013scaling} still assume that the channel state information (CSI) is available for signal detection at the receivers.

More recently, a new proposal has emerged, namely, the mixture of blind deconvolution and demixing \cite{ling2015blind}, i.e., \emph{blind demixing} for brief, regarded as a promising solution to support the efficient low-latency communication without channel estimation at both transmitters and receivers. It also meets the demands for sporadic and short messages in next generation wireless networks \cite{jung2017blind}. In particular, blind deconvolution is a problem of estimating two unknown vectors from their convolution, which can be exploited in the context of channel coding for multipath channel protection \cite{ahmed2014blind}. However, the results of the blind deconvolution problem \cite{ahmed2014blind} cannot be directly extended to the blind demixing problem since only a single observed vector is available. 
Demixing refers to the problem of identifying multiple structured signals by given the mixture of measurements of these signals, which can be exploited in a secure communications protocol \cite{mccoy2014sharp}. The measurement matrices in the demixing problem are normally assumed to be full-rank matrices to assist theoretical analysis \cite{mccoy2013achievable}. However, the measurement matrices in the blind demixing problem are rank-one matrices \cite{ling2017regularized}, which hamper the extension of results developed in \cite{mccoy2013achievable} to the blind demixing problem.

In this paper, we consider the blind demixing problem in a specific scenario, i.e., an orthogonal frequency
division multiplexing (OFDM) system and propose a low-rank approach to recover the original signals in this problem. However, the resulting rank-constrained optimization problem is known to be non-convex and highly intractable. A growing body of literature has proposed marvelous algorithms to deal with low-rank problems. In particular, convex relaxation approach is an effective way to solve this problem with theoretical guarantees \cite{ling2015blind}. However, it is not scalable to the medium- and large-scale problems due to the high iteration cost of the convex programming technique. To enable scalability, non-convex algorithms (e.g, regularized gradient descent algorithm \cite{ling2017regularized} and iterative hard thresholding method \cite{strohmer2017painless}), endowed with lower iteration cost, have been developed. However, the overall computational complexity of these algorithms is still high due to the slow convergence rate, i.e., high iteration complexity.

To address the limitations of the existing algorithms for the blind demixing problem, we propose the Riemannian optimization algorithm over a product complex manifold in order to simultaneously reduce the iteration cost and iteration complexity. Specifically, the quotient manifold geometry of the product of complex symmetric rank-one matrices is exploited. This is achieved by equivalently reformulating the original complex asymmetric matrices as Hermitian positive semidefinite (PSD) matrices. To reduce the iteration complexity, the Riemannian gradient descent algorithm and the Riemannian trust-region algorithm are developed to support linear convergence rate and superlinear convergence rate, respectively. By exploiting the benign geometric structure of the blind demixing problem, i.e., symmetric rank-one matrices, the iteration cost can be significantly reduced (the same as the regularized gradient descent algorithm \cite{ling2017regularized}) and is scalable to large-size problem.

In this paper, we prove that, for blind demixing, the Riemannian gradient descent algorithm with spectral initialization can linearly converge to the ground truth signals with high probability provided sufficient measurements. Numerical experiments will demonstrate the algorithmic advantages, performance gains and sample efficiency of the Riemannian optimization algorithms.

\subsection{Related Works}
\subsubsection{Convex Optimization Approach}
To address the algorithmic challenge of the rank-constraint optimization problem, the work \cite{ling2015blind} investigated the nuclear norm minimization method for the blind demixing problem. Although the nuclear norm based approach can solve this problem in polynomial time, the high iteration complexity yielded by the resulting semidefinite program (SDP) limits the scalability of the convex relaxation approach. This motivates the development of non-convex algorithms in order to reduce the iteration cost and simultaneously maintain competitive theoretical recovery guarantees compared with convex methods.
\subsubsection{Non-Convex Optimization Paradigms}
The recent work \cite{ling2017regularized} proposed a non-convex regularized gradient-descent based method with a elegant initialization to solve the blind demixing problem at a linear convergence rate. Another work \cite{strohmer2017painless} implemented thresholding-based methods to solve the demixing problem for general rank-$r$ matrices. The algorithm in \cite{strohmer2017painless} linearly converges to the global minimum with a similar initial strategy in \cite{ling2017regularized}. Even though the iteration cost of the non-convex algorithm is lower than the convex approach, the overall computational complexity is still high due to the slow convergence rate, i.e., high iteration complexity. Moreover, both non-convex algorithms require careful initialization to achieve desirable performance.

To address the above limitations of the existing algorithms, we develop the Riemannian optimization algorithms to simultaneously reduce
the iteration cost and iteration complexity.
However, most of current developed Riemannian optimization algorithms for low-rank optimization problem \cite{shi2016low,boumal2011rtrmc} are developed in real space with respect to a single optimization variable. Recently, a Riemannian steepest descent method is developed to solve the blind deconvolution problem \cite{huang2017blind}, where the regularization is needed to provide statistical guarantees. In the blind demixing problem, the following coupled challenges arise due to \emph{multiple complex asymmetric} variables:
\begin{itemize}
	\item Constructing product Riemannian manifold for the multiple complex asymmetric rank-one matrices.
	\item Developing the Riemannian optimization algorithm on the complex product manifold.
\end{itemize} 
Therefore, it is crucial to address these unique challenges to solve the blind demixing problem via the Riemannian optimization technique.
\subsection{Contributions}
The major contributions of the paper are summarized as follows:
\begin{enumerate}
	\item We present a novel blind demixing approach to support low-latency communication in an orthogonal frequency
	division multiplexing (OFDM) system, thereby recovering the information signals without channel estimation. A low-rank approach is further developed to solve the blind demixing problem. 
	\item To efficiently exploit the quotient manifold geometry of the product of complex symmetric rank-one matrices, we equivalently reformulate the original complex asymmetric matrices to Hermitian positive semidefinite matrices.
	\item To simultaneously reduce the iteration cost and iteration complexity as well as enhance the estimation performance, we develop the scalable Riemannian gradient descent algorithm and Riemannian   trust-region algorithm by exploiting the benign geometric structure of symmetric rank-one matrices. This is achieved by factorizing the symmetric rank-one matrices.
	\item  We prove that, for blind demixing, the Riemannian gradient descent linearly converges to the ground truth signals with high probability provided sufficient measurements. 
\end{enumerate}
 
 \subsection{Organization and Notations}
 The remainder of this paper is organized as follows. The system model and problem formulations are presented in Section \ref{2}. In Section \ref{3}, we introduce the versatile framework of Riemannian optimization on the product manifold. The process of computing optimization related ingredients and the Riemannian optimization algorithms are explicated in Section \ref{4.0}. The theoretical guarantees of the Riemannian gradient descent algorithm is then presented in Section \ref{section:MT}. Numerical results will be illustrated in Section \ref{4}. We further conclude this paper in Section \ref{5}.
 
 Throughout the paper following notions are used. Vectors and matrices are denoted by bold lowercase and bold uppercase letters respectively. Specifically, we let $\{\bm{X}_k\}_{k=1}^s$ be the target matrices and $\{\bm{X}_k^{[t]}\}_{k=1}^s$ be the $t$-th iterate of the algorithms where $s$ is the number of users/devices.
 For a vector $\bm{z}$, $\|\bm{z}\|$ denotes its Euclidean norm. For a matrix $\bm{M}$, $\|\bm{M}\|_*$ and $\|\bm{M}\|_F$ denote its nuclear norm and Frobenius norm respectively. For both matrices and vectors, $\bm{M}^{\mathsf{H}}$ and $\bm{z}^{\mathsf{H}}$ denote their complex conjugate transpose. $\bar{{\bm{z}}}$ is the complex conjugate of the vector $\bm{z}$. $a^*$ is the complex conjugate of the complex constant $a$. The inner product of two complex matrices $\bm{M}_1$ and $\bm{M}_2$ is defined as $\langle\bm{M}_1,\bm{M}_2\rangle = \trace(\bm{M}_1^{\mathsf{H}}\bm{M}_2)$. Let $\bm{I}_L$ denote the identity matrix with size of $L\times L$.

\section{System Model And Problem Formulation}\label{2}	
In this section, we present a blind demixing approach to support the low-latency communication by reducing the channel signaling overhead in orthogonal frequency
division multiplexing (OFDM) system. We develop a low-rank optimization model to recover the original information signal for the blind demixing problem, which, however, turns out to be highly intractable. The Riemannian optimization approach is then motivated to address the computational issue.
\begin{figure}[tb]
	\centering
	\includegraphics[width=0.5\columnwidth]{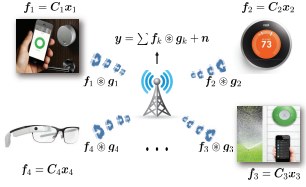}
	\caption{Blind demixing for multi-user low-latency communication systems	without channel estimation.}
	\label{fig:sys}
\end{figure}
\subsection{Orthogonal Frequency Division Multiplexing}
We first briefly introduce the orthogonal frequency division multiplexing (OFDM). The basic idea of OFDM is that by exploiting the eigenfunction property of sinusoids in linear-time-variant (LTI) system, transformation into the frequency domain is a particularly benign way to communicate over frequency-selective channels \cite{tse2005fundamentals}. 

Let $p[n]$ and $\theta[n]$ denote the transmitted  signal and the received signal in the $n$-th time slot, respectively. $q_{\ell}$ denotes the $\ell$-th tap channel impulse response which does not change with $n$. Thus the channel is linear time-invariant. Here, the discrete-time model is given as
\setlength\arraycolsep{2pt} 
\begin{align}
{\theta}[n] = \sum_{\ell = 0}^{L^\prime-1}{q}_{\ell}{p}[n-\ell],
\end{align}
where $L^\prime$ is a finite number of non-zero taps. However, the sinusoids are not eigenfunctions when we transmit data symbols $\bm{p} = [p[0],p[1],\cdots,p[N_p-1]]^{\top}\in\mathbb{C}^{N_p}$ over only a finite duration \cite{tse2005fundamentals}. To restore the eigenfunction property, we add a \emph{cyclic prefix} to $\bm{p}$.  Specifically, we add a prefix of length $L^\prime-1$ consisting of data symbols rotated cyclically:
\begin{align}
\bm{d} = [{p}[N_p-L^\prime+1],\cdots,{p}[N_p-1],{p}[0],{p}[1],\cdots,{p}[N_p-1]]^{\top}\in\mathbb{C}^{N_p+L^\prime-1}.
\end{align}
We only consider the output over the time interval $n\in[L^\prime,N_p+L^\prime-1]$, represented as
\begin{align}\label{linearconv}
{\theta}[n] = \sum_{\ell = 0}^{L^\prime-1}{q}_{\ell}{d}[(n-L^\prime-\ell)~\text{modulo}~ N_p].
\end{align}
Denoting the output of length $N_p$ as
\begin{align}
\bm{\theta} = [\theta[L^\prime],\cdots,\theta[N_p+L^\prime-1]]^{\top},
\end{align}
and the channel impulse as 
$
\bm{q} = [q_0,q_1,\cdots, q_{L^\prime-1},0,\cdots,0]^{\top}\in\mathbb{C}^{N_p}$, thus
 (\ref{linearconv}) can be rewritten as
$\bm{\theta} = \bm{q}\circledast\bm{p},
$
where the notion $\circledast$ denotes the \emph{cyclic convolution}. This method is called \emph{cyclic extension}.
\subsection{System Model}
In orthogonal frequency division multiplexing, we consider a network with one base station and $s$ mobile users, as shown in Fig. \ref{fig:sys}. Specifically, let $\bm{x}_k\in\mathbb{C}^N$ be the original signals of length $N$ from the $k$-th user,                                                                                                                                                                                                                                                                                                 which is usually assumed to be drawn from Gaussian distribution for the convenience of theoretical analysis \cite{ling2015blind,ling2017regularized}. To make the model practical, in contrast, we consider $\bm{x}_k$ as an OFDM signal \cite[Section 2.1]{beek1998synchronization}, which consists of $N$ orthogonal subcarriers modulated by $N$ parallel data streams, represented as
\begin{align}\label{x}
\bm{x} = \bm{F}^{\mathsf{H}}\bm{s},
\end{align} 
where $\bm{F}\in\mathbb{C}^{N\times N}$ is the normalized discrete Fourier transform (DFT) matrix and $\bm{s}\in\mathbb{C}^{N}$ is taken from QAM symbol constellation.
Over $L$ time slots, the transmit signal at the $k$-th user is given by
$
\bm{f}_k = \bm{C}_k\bm{x}_k,
$
where $\bm{C}_k\in\mathbb{C}^{L\times N}$ with $L>N$ is the encoding matrix and available to the base station.
The signals $\bm{f}_k$'s are passing through individual time-invariant channels with impulse responses $\bm{h}_k$'s where $\bm{h}_k\in\mathbb{C}^K$ has a maximum delay of at most $K$ samples. In the OFDM system, this model can be represented in terms of circular convolutions of transmit signals $\bm{f}_k\in\mathbb{C}^L$ with zero-padded channel vector $\bm{g}_k\in\mathbb{C}^L$, i.e.,
\begin{align}
\bm{g}_k = [\bm{h}_k^{\top},0,\cdots,0]^{\top}. 
\end{align}
Hence, the received signal is given as
\begin{align}\label{y}
\bm{z} = \sum\nolimits_{k=1}^{s}\! \bm{f}_k\circledast \bm{g}_k+\bm{n},
\end{align}
where $\bm{n}$ denotes the additive white complex Gaussian noise. Our goal is to recover the original information signals $\{\bm{x}_k\}_{k=1}^s$ from the single observation $\bm{z}$ without knowing channel impulse response $\{\bm{g}_k\}_{k=1}^s$. We call this problem as the \emph{blind demixing} problem.

However, the above information recovery problem is highly intractable without any further structural assumptions. Fortunately, in wireless communication, we can design the encoding matrices $\{\bm{C}_k\}_{k=1}^s$ such that it satisfies ``local" mutual incoherence conditions \cite{ling2015blind}. Specifically, from the practical points of view, we design the encoding matrix $\bm{C}_k$ as a Hadamard-type matrix \cite{ling2017regularized}, represented as
\begin{align}\label{C}
 \bm{C}_k = \bm{F}\bm{D}_k\bm{H},
\end{align}
 where $\bm{F}\in \mathbb{C}^{L\times L}$ is the normalized discrete Fourier transform (DFT) matrix, $\bm{D}_k$'s are independent diagonal binary $\pm 1$ matrices and $\bm{H}\in\mathbb{C}^{L\times N}$ fixed partial deterministic Hadamard matrix.  Furthermore,  due to the physical properties of channel propagation \cite{goldsmith2005wireless}, the impulse response $\bm{g}_k$ is compactly supported \cite{strohmer2017painless}. Here, the size of the compact set of $\bm{g}_k$, i.e., $K$  where $K<L$, is termed as its maximum delay spread from an engineering perspective \cite{strohmer2017painless}. In this paper, we assume that the impulse response $\bm{g}_k$ is not available to both receivers and transmitters during the transmissions in order to reduce the channel signaling overhead \cite{ling2017regularized}.

\subsection{Demixing of Rank-One Matrices}
Let $\bm{B}\in\mathbb{C}^{L\times K}$ consist of the first $K$ columns of $\bm{F}$. Due to the benign property of cyclic extension, it is convenient to represent the formulation (\ref{y}) in the Fourier domain \cite{ling2015blind,tse2005fundamentals}
\begin{align}\label{fdomain}
\bm{y} =\bm{Fz} = \sum_k (\bm{FC}_k\bm{x}_k)\odot\bm{Bh}_k+\bm{F}\bm{n},
\end{align}
where $\odot$ denotes the componentwise product. The first term of (\ref{fdomain}) can be further rewritten as \cite{strohmer2017painless}
$
[(\bm{FC}_k\bm{x}_k)\odot\bm{Bh}_k]_i = (\bm{c}_{ki}^{\mathsf{H}}\bm{x}_k)(\bm{b}_i^{\mathsf{H}}\bm{h}_k) =\langle \bm{c}_{ki}\bar{\bm{b}}_{i}^{\mathsf{H}},\bm{X}_k\rangle,
$
where $\bm{c}^{\mathsf{H}}_{ki} $ denotes the $i$-th row of $\bm{FC}_k$, $\bm{b}_{i}^{\mathsf{H}} $ represents the $i$-th row of $\bm{B}$ and $\bm{X}_k = \bm{x}_k\bar{\bm{h}}_k^{\mathsf{H}}\in\mathbb{C}^{N\times K}$ is a rank-one matrix. Hence, the received signal at the base station can be represented in the Fourier domain as
\begin{align}\label{equ}
\bm{y} =\sum\nolimits_{k = 1}^{s}\! \mathcal{A}_k(\bm{X}_k)+\bm{e},
\end{align}
where the vector $\bm{e} = \bm{Fn}$ and the linear operator $\mathcal{A}_k:\mathbb{C}^{N\times K}\rightarrow\mathbb{C}^L$ is given as \cite{ling2017regularized}
\begin{align}\label{linear_op}
\mathcal{A}_k(\bm{X}_k):=\{\langle \bm{c}_{ki} \bar{\bm{b}}_i ^{\mathsf{H}},\bm{X}_k\rangle\}_{i=1}^L = \{\langle \bm{A}_{ki} ,\bm{X}_k\rangle\}_{i=1}^L,
\end{align}
with $\bm{A}_{ki}=\bm{c}_{ki} \bar{\bm{b}}_i ^{\mathsf{H}}$.
We thus formulate the blind demixing problem as the following low-rank optimization problem:
\begin{eqnarray}
\label{lea_squ}
\mathscr{P}:\mathop{\textrm{minimize }}_{\bm{W}_k,k=1,\cdots,s}&& \Big\|\sum\nolimits_{k=1}^{s}\!\mathcal{A}_k(\bm{W}_k)-\bm{y}\Big\|^2
\nonumber \\
\textrm{subject to}&& \textrm{rank}(\bm{W}_k)= 1,~{k=1,\cdots,s},
\end{eqnarray}
where $\bm{W}_k\in\mathbb{C}^{N\times K}, k = 1,\cdots,s$ and $\{\mathcal{A}_k\}_{k=1}^s$ are known. However, problem $\mathscr{P}$ turns out to be highly intractable due to non-convexity of rank-one constraints. Despite the non-convexity of problem (\ref{lea_squ}), we solve this problem via exploiting the quotient manifold geometry
of the product of complex symmetric rank-one matrices. This is achieved by equivalently reformulating the original complex
asymmetric matrices as the Hermitian positive semidefinite (PSD)
matrices. The first-order and second-order algorithms are further developed on the quotient manifold, which enjoy algorithmic advantages, performance gains and sample efficiency.

\subsection{Problem Analysis}
 To address the algorithmic challenge of problem $\mathscr{P}$, enormous progress has been recently made to develop convex methods \cite{ling2015blind,mccoy2013achievable} and non-convex methods \cite{ling2017regularized,strohmer2017painless}. In this subsection, we will first review the existing algorithms for the blind demixing problem. Then we identify the limitations of state-of-the-art algorithms and develop Riemannian trust-region algorithm to address these limitations. Unique challenges of developing the Riemannian optimization algorithm will be further revealed.
\subsubsection{Convex Relaxation Approach}
A line of literature \cite{ling2015blind} adopted the nuclear norm minimization method to reformulate the problem $\mathscr{P}$ as
\begin{eqnarray}
\label{convex_re}
\mathop{\textrm{minimize }}_{\bm{W}_k,k=1,\cdots,s}&& \sum\nolimits_{k = 1}^{s}\|\bm{W}_k\|_*
\nonumber \\
\textrm{subject to}&& \Big\|\sum\nolimits_{k = 1}^{s}\!\mathcal{A}_k(\bm{W}_k)-\bm{y}\Big\|\leq\varepsilon,
\end{eqnarray}
where $\bm{W}_k\in\mathbb{C}^{N\times K}$ and the parameter $\varepsilon$ is an upper bound of $\|\bm{e}\|$ in (\ref{equ}) and assumed to be known.
While the blind demixing problem can be solved by convex technique provably and robustly under certain situations, the convex relaxation approach is computationally infeasible to the medium-scale or large-scale problems due to the limitations of high iteration cost. This motivates the development of efficient non-convex approaches with lower iteration cost.
\subsubsection{Non-convex Optimization Paradigms}
A line of recent work \cite{ling2017regularized,strohmer2017painless} has developed non-convex algorithms which reduces the iteration cost. In particular, work \cite{strohmer2017painless} solved problem $\mathscr{P}$  via the hard thresholding technique. Specifically, the $t$-th iterate with respect to the $k$-th variable is given by
$
\bm{W}_k^{[t+1]} = \mathcal{F}_r\big(\bm{W}_k^{[t]}+\alpha_k^{[t]}\mathcal{P}_{T_{k,t}}(\bm{G}_k^{[t]})\big),
$
where the hard thresholding operator $\mathcal{F}_r$ returns the best rank-$r$ approximation of a matrix, $\mathcal{P}_{T_{k,t}}(\bm{G}_k^{[t]})$ represents the projection of the search direction to the tangent space $T_{k,t}$, and the stepsize is denoted as
$
\alpha_k^{[t]} ={\|\mathcal{P}_{T_{k,t}}(\bm{G}_k^{[t]})\|_F^2}/{\|\mathcal{A}_k\mathcal{P}_{T_{k,t}}(\bm{G}_k^{[t]})\|_2^2}.
$  
Therein,  $\bm{G}_k^{[t]}$ is defined in \cite{strohmer2017painless}, given by $\bm{G}_k^{[t]} = \mathcal{A}_k^*(\bm{r}^{[t]}),$where $
\bm{r}^{[t]} = \bm{y} - \sum_{k=1}^s \mathcal{A}_k(\bm{W}_k^{[t]}),
\mathcal{A}_k^*(\bm{z}) =\sum_{l=1}^L z_l{\bm{b}}_i \bm{c}_{ki}^{\mathsf{H}}.$

Moreover, matrix factorization also serves as a powerful method to address the low-rank optimization problem. Specifically, \cite{ling2017regularized}  developed an algorithm solving the blind demixing problem based on matrix factorization and regularized gradient descent method. Specifically, problem $\mathscr{P}$ can be rewritten as
\begin{eqnarray}
\label{non-convex}
\mathop{\textrm{minimize }}_{{\bm{u}_k,\bm{v}_k,k=1,\cdots,s}} 
F(\bm{u},\bm{v}):= g(\bm{u},\bm{v})+\lambda R(\bm{u},\bm{v}),
\end{eqnarray}
where $g(\bm{u},\bm{v}):=\|\sum\nolimits_{k=1}^s\mathcal{A}_k(\bm{u}_k\bm{v}_k^{\mathsf{H}})-\bm{y}\|^2$ with $\bm{u}_k\in\mathbb{C}^{N},\bm{v}_k\in\mathbb{C}^{K}$ and the regularizer $R(\bm{u},\bm{v})$ is proposed to force the iterates to lie in the \emph{basin of attraction} \cite{ling2017regularized}.
The algorithm starts from a good initial point and updates the iterates simultaneously: 
\begin{align}
\bm{u}_k^{[t+1]} &= \bm{u}_k^{[t]}-\eta\nabla F_{\bm{u}_k}(\bm{u}_k^{[t]},\bm{v}_k^{[t]}),\quad\\
\bm{v}_k^{[t+1]} &= \bm{v}_k^{[t]}-\eta\nabla F_{\bm{v}_k}(\bm{u}_k^{[t]},\bm{v}_k^{[t]}),
\end{align}
where $\nabla F_{\bm{u}_k}$ denotes the derivative of the objective function (\ref{non-convex}) with respect to $\bm{u}_k$. Although the above non-convex algorithms have low iteration cost, the overall computational complexity is still high due to the slow convergence rate, i.e., high iteration complexity. This motivates to design efficient algorithms to simultaneously reduce the iteration cost and iteration complexity. 
\subsection{Riemannian Optimization Approach}
In this paper, we develop Riemannian optimization algorithms to solve problem $\mathscr{P}$, thereby addressing the limitations of the existing algorithms (e.g., regularized gradient descent algorithm \cite{ling2017regularized}, nuclear norm minimization method \cite{ahmed2014blind} and fast iterative hard thresholding algorithm \cite{strohmer2017painless}) by
\begin{itemize}
	\item Exploiting the Riemannian
	quotient geometry of the product of \emph{complex asymmetric rank-one} matrices to reduce the iteration cost.
	\item Developing scalable Riemannian gradient descent algorithm and Riemannian trust-region algorithm to reduce the iteration complexity.
\end{itemize}

The Riemannian optimization technique has been applied in a wide range of areas to solve rank-constrained problem and achieves excellent performance, e.g., the low-rank matrix completion problem \cite{boumal2011rtrmc,vandereycken2013low,mishra2014fixed}, topological interference management problem \cite{shi2016low} and blind deconvolution \cite{huang2017blind}. However, all of the current Riemannian optimization problems for low-rank optimization problem are developed on the real Riemannian manifold (e.g., real Grassmann manifold \cite{boumal2011rtrmc} and quotient manifold of fixed-rank matrices \cite{shi2016low,vandereycken2013low,mishra2014fixed}) with single non-symmetric variable or the complex Riemannian manifold with single symmetric variable \cite{Yatawatta2013A,Yatawatta2013B} or the complex Riemannian manifold with single non-symmetric variable \cite{huang2017blind}. Thus unique challenges arise due to \emph{multiple complex asymmetric variables} in problem $\mathscr{P}$. In this paper, we propose to construct product Riemannian manifold for the multiple complex asymmetric rank-one matrices, followed by developing the scalable Riemannian optimization algorithms.
 
\section{Riemannian Optimization over Product Manifolds}\label{3}
In this section, to exploit the Riemannian quotient geometry of the product of complex symmetric rank-one matrices \cite{Yatawatta2013A}, we reformulate the original optimization problem on complex asymmetric matrices to the one on Hermitian positive semidefinite (PSD) matrices. The Riemannian optimization algorithms are further developed via exploiting the quotient manifold geometry of the complex product manifold.
\subsection{Product Riemannian Manifold}
Let the manifold $\mathcal{M}$ denote the Riemannian manifold endowed with the Riemannian metric $\bm{g}_{\bm{M}_k}$ where $k\in[s]$ with $[n]=\{1,2,\cdots,n\}$. The set
$
\mathcal{M}^s = \substack{\underbrace{\mathcal{M}\times \mathcal{M}\times \cdots\times \mathcal{M}}\\s}
$ is defined as the set of matrices $(\bm{M}_1,\cdots,\bm{M}_s)$ where $\bm{M}_k\in\mathcal{M},k=1,2\cdots,s$, and is called product manifold. Its manifold topology is identified to the product topology \cite{absil2009optimization}. 

In order to develop the optimization algorithms on the manifold, a notion of length that applies to tangent vectors is needed \cite{absil2009optimization}. By taking the manifold $\mathcal{M}$ as an example, this goal is achieved by endowing each tangent space $T_{\bm{M}}\mathcal{M}$ with a smoothly varying inner product $g_{\bm{M}}(\bm{\zeta}_{\bm{M}},\bm{\eta}_{\bm{M}})$ where $\bm{\zeta}_{\bm{M}},\bm{\eta}_{\bm{M}}\in T_{\bm{M}}\mathcal{M}$.
Endowed with a inner product $g_{\bm{M}}$, the manifold $\mathcal{M}$ is called the \emph{Riemannian manifold}. The smoothly varying inner product is called the \emph{Riemannian metric}. The above discussions are also applied to the product manifold.
Based on the above discussions, we characterize the notion of length on the product manifold via endowing tangent space $T_{\bm{V}}\mathcal{M}^s$ with the smoothly varying inner product
\begin{align}\label{vmetric}
g_{\bm{V}}(\bm{\zeta_V},\bm{\eta_V}):=\sum\nolimits_{k=1}^{s}g_{\bm{M}_k}(\bm{\zeta}_{\bm{M}_k},\bm{\eta}_{\bm{M}_k}). 
\end{align} 
Thus, with $\mathcal{M}$ as the Riemannian manifold, the product manifold $\mathcal{M}^s$ is also called Riemannian manifold, endowed with the Riemannian metric $g_{\bm{V}}$. 
\subsection {Handling Complex Asymmetric Matrices}
To handle complex asymmetric matrices, we propose a linear map which is exploited to convert the optimization variables to the Hermitian positive semidefinite matrix.  Let $\mathbb{S}^{(N+K)}_+$ denote the set of Hermitian positive semidefinite matrices.
Define the linear map $\mathcal{J}_k:\mathbb{S}^{(N+K)}_+\rightarrow\mathbb{C}^L$ and a Hermitian positive semidefinite (PSD) matrix $\bm{Y}_k$ such that $[\mathcal{J}_k(\bm{Y}_k)]_i =  \langle \bm{J}_{ki},\bm{Y}_k\rangle$ with $\bm{Y}_k\in\mathbb{S}_+^{(N+K)}$ and $\bm{J}_{ki}$ as
\begin{align}\label{B}
\bm{J}_{ki} =\left[{\bm{0}_{N\times N} \atop\bm{0}_{K\times N}} {\bm{A}_{ki}\atop \bm{0}_{K\times K}}\right]\in\mathbb{C}^{(N+K)\times (N+K)},
\end{align}where $\bm{A}_{ki}$ is given in (\ref{linear_op}).
Note that 
\begin{align}\label{J_M}
[\mathcal{J}_k(\bm{M}_k)]_i =\langle \bm{J}_{ki},\bm{M}_k\rangle = \langle \bm{A}_{ki},\bm{x}_k\overline{\bm{h}}_k^{\mathsf{H}}\rangle
\end{align} 
where $\bm{M}_k = \bm{w}_k\bm{w}_k^{\mathsf{H}}$ with $\bm{w}_k=[\begin{matrix}
\bm{x}_k^{\mathsf{H}}~
\overline{\bm{h}}_k^{\mathsf{H}}
\end{matrix}
]^{\mathsf{H}}\in\mathbb{C}^{N+K}$. Hence, problem $\mathscr{P}$ can be equivalently reformulated as the following optimization problem on the set of Hermitian positive semidefinite matrices
\begin{eqnarray}
\label{lea_squ_sdp}
\mathop{\textrm{minimize }}_{\bm{M}_k,{k = 1,\cdots,s}}&& \Big\|\sum\nolimits_{k=1}^{s}\!\mathcal{J}_k(\bm{M}_k)-\bm{y}\Big\|^2
\nonumber \\
\textrm{subject to}&& \textrm{rank}(\bm{M}_k)= 1,~{k=1,\cdots,s},
\end{eqnarray}
where $\bm{M}_k\in\mathbb{S}^{(N+K)}_+$. Based on equality (\ref{J_M}), we know that the top-right $N\times K$ submatrix of the estimated matrix $\hat{\bm{M}}_k$ in problem (\ref{lea_squ_sdp}) is corresponding to the estimated matrix $\hat{\bm{W}}_k$ in problem (\ref{lea_squ}). Furthermore, we define $\bm{V}=\{\bm{M}_k\}_{k=1}^s\in\mathcal{M}^s$, where
$\mathcal{M}$ denotes the manifold encoded with the rank-one
matrices and $\mathcal{M}^s$ represents the product of $s$ manifolds $\mathcal{M}$.
By exploiting the Riemannian manifold geometry, the
rank-constrained optimization problem (\ref{lea_squ_sdp}) can be transformed
into the following unconstrained optimization problem over the search space of the product manifold $\mathcal{M}^s$:
\begin{eqnarray}\label{manifold_opt}
\mathop{\textrm{minimize }}_{\bm{V}=\{\bm{M}_k\}_{k=1}^s} f(\bm{V}):=\Big\|\sum\nolimits_{k = 1}^{s}\!\mathcal{J}_k( \bm{M}_k)-\bm{y}\Big\|^2.
\end{eqnarray}
Note that the theoretical advantages of the symmetric transformation have been recently revealed in \cite{ge2017no} for the low-rank optimization problems in machine learning and high-dimensional statistics. Furthermore, by factorizing the Hermitian positive semidefinite matrix $\bm{M}_k = \bm{w}_k\bm{w}_k^{\mathsf{H}}$, $k = 1,\cdots,s$, it yields only $s$ vector variables, i.e., $\bm{w}_k\in\mathbb{ C}^{N+K}$, for optimization problem (\ref{manifold_opt}), which simplifies the derivation of optimization related ingredients.

\subsection{Quotient Manifold Space}
The main idea of Riemannian optimization for rank-constrained problem is based on matrix factorization \cite{mishra2014fixed,Yatawatta2013A}. In particular, the factorization $\bm{M}_k= \bm{w}_k\bm{w}_k^{\mathsf{H}}$ where $\bm{w}_k\in\mathbb{C}^{N+K}$ is prevalent in dealing with rank-one Hermitian positive semidefinite matrices \cite{Yatawatta2013A,Yatawatta2013B}. This factorization also takes advantages of lower-dimensional search space \cite{boumal2011rtrmc} over the other general forms of matrix factorization for rank-one matrices \cite{mishra2014fixed}.
However, the factorization $\bm{M}_k = \bm{w}_k\bm{w}_k^{\mathsf{H}}$ is not unique because the transformation $\bm{w}_k\mapsto a_k\bm{w}_k$ leaves the matrix $\bm{w}_k\bm{w}_k^{\mathsf{H}}$ unchanged, where $a_k\in\mathrm{SU}(1):=\{a_k\in\mathbb{C}: a_ka_k^{*} = a_k^{*}a_k = 1\}$ and $\mathrm{SU}(1)$ is the special unitary group of degree $1$. In particular, the non-uniqueness yields a profound affect on the performance of second-order optimization algorithms which require non-degenerate critical points. To address this indeterminacy, we encode the transformation $\bm{w}_k\mapsto a_k\bm{w}_k$ where $k=1,2,\cdots ,s$, in an abstract search space to construct the equivalence class:
 \begin{align}\label{equ2}
[\bm{M}_k]= \{a_k\bm{w}_k:a_ka_k^{*} = a_k^{*}a_k = 1, a_k\in\mathbb{C}\}.
\end{align}
The product of $[\bm{M}_k]$'s yields the equivalence class
\begin{align}\label{equclass}
	[\bm{V}]  &= \{[\bm{M}_k]\}_{k=1}^s.
\end{align}  
The set of equivalence classes \eqref{equclass} is denoted as $\mathcal{M}^s/\sim$, called the \emph{quotient space} \cite{Yatawatta2013A}. Since the quotient manifold $\mathcal{M}^s/\sim$ is an abstract space, in order to execute the optimization algorithm, the matrix representations defined in the computational space are needed to represent corresponding abstract geometric objects in the abstract space \cite{absil2009optimization}. We denote an element of the quotient space $\mathcal{M}^s/\sim$ by $\tilde{\bm{V}}$ and its matrix representation in the computational space $\mathcal{M}^s$ by $\bm{V}$. Therefore, there is $\tilde{\bm{V}} = \pi(\bm{V})$ and $[\bm{V}] = \pi^{-1}(\pi({\bm{V}}))$, where the mapping $\pi:\mathcal{M}^s\rightarrow\mathcal{M}^s/\sim$ is called the natural projection constructing the map between the geometric ingredients of the computational space and the ones of the quotient space.
\subsection{The Framework of Riemannian Optimization}
To develop the Riemannian optimization algorithms over the quotient space, the geometric concepts in the abstract space $\mathcal{M}^s/\sim$ call for the matrix representations in the computational space $\mathcal{M}^s$. Specifically, several classical geometric concepts in the Euclidean space are required to be generalized to the geometric concepts on the manifold, such as the notion of length to the Riemannian metric, set of directional derivatives to the tangent space and motion along geodesics to  the retraction. We first present the general Riemannian optimization developed on the product manifold in this subsection. The details on the derivation of concrete optimization-related ingredients will be introduced in the next section.

Based on the Riemannian metric (\ref{vmetric}), the tangent space $T_{\bm{V}}\mathcal{M}^s$ can be decomposed into two complementary vector spaces, given as \cite{absil2009optimization}
\begin{align}\label{complen}
T_{\bm{V}}\mathcal{M}^s = \mathcal{V}_{\bm{V}}\mathcal{M}^s\oplus\mathcal{H}_{\bm{V}}\mathcal{M}^s,
\end{align}
where $\oplus$ is the direct sum operator. Here, $\mathcal{V}_{\bm{V}}\mathcal{M}^s$ is the \emph{vertical space} where directions of vectors are tangent to the set of equivalence class (\ref{equclass}). $\mathcal{H}_{\bm{V}}\mathcal{M}^s$ is the \emph{horizontal space} where the directions of vectors are orthogonal to the equivalence class (\ref{equclass}). Thus the tangent space $T_{\tilde{\bm{V}}}(\mathcal{M}^s/\sim)$ at the point $\tilde{\bm{V}}\in \mathcal{M}^s/\sim$ can be represented by the horizontal space $\mathcal{H}_{\bm{V}}\mathcal{M}^s$ at point $\bm{V}\in \mathcal{M}^s$. In particular, the matrix representation of $\eta_{\tilde{\bm{V}}}\in T_{{\tilde{\bm{V}}}}(\mathcal{M}^s/\sim)$, which is also called the \emph{horizontal lift} of $\eta_{\tilde{\bm{V}}}\in T_{{\tilde{\bm{V}}}}(\mathcal{M}^s/\sim)$ \cite[Section 3.5.8]{absil2009optimization}, can be represented by a unique element $\bm{\eta}_{\bm{V}}\in\mathcal{H}_{\bm{V}}\mathcal{M}^s$.
In addition, for every $\bm{\xi}_{\bm{V}},\bm{\eta}_{\bm{V}}\in T_{\bm{V}}\mathcal{M}^s$, based on the Riemannian metric $g_{\bm{V}}({\bm{\zeta_V}},{\bm{\eta_V}})$ (\ref{vmetric}),
\begin{align}\label{metric_q}
g_{\tilde{\bm{V}}}(\bm{\zeta}_{\tilde{\bm{V}}},\bm{\eta}_{\tilde{\bm{V}}}):=  g_{\bm{V}}(\bm{\zeta}_{\bm{V}},\bm{\eta}_{\bm{V}})
\end{align}
defines a Rimannian metric on the quotient space $\mathcal{M}^s/\sim$ \cite[Section 3.6.2]{absil2009optimization}, where $\bm{\zeta}_{\tilde{\bm{V}}},\bm{\eta}_{\tilde{\bm{V}}}\in T_{\tilde{\bm{V}}}\mathcal{M}^s$. Endowed with the Riemannian metric (\ref{metric_q}), the natural projection map $\pi:\mathcal{M}^s\rightarrow\mathcal{M}^s/\sim$ is a \emph{Riemannian submersion} from the quotient manifold $\mathcal{M}^s/\sim$ to the computational space $\mathcal{M}^s$  \cite[Section 3.6.2]{absil2009optimization}. 
Based on the Riemannian submersion theory, several objects on the quotient manifold can be represented by corresponding objects in the computational space.

 Based on the aforementioned framework, the Riemannian optimization over the computational space $\mathcal{M}^s$ can be briefly described as follows. First, search the directions $\bm{\eta}_{\bm{V}}$ on the horizontal space $\mathcal{H}_{\bm{V}}\mathcal{M}^s$. Then, retract the directional vector $\bm{\eta}_{\bm{V}}$ onto the manifold $\mathcal{M}^s$ via the mapping $\mathcal{R}_{\bm{V}}:\mathcal{H}_{\bm{V}}\mathcal{M}^s\rightarrow\mathcal{M}^s$ called \emph{retraction}.  Since the manifold topology of the product manifold is equivalent to the product topology \cite[Section 3.1.6]{absil2009optimization}, i.e., the aforementioned optimization operations can be handled individually, the Riemannian optimization developed on product manifold $\mathcal{M}^s$ can be individually processed on the manifold $\mathcal{M}$.

 Specifically, the tangent space $T_{\bm{V}}\mathcal{M}^s$ can be termed as the product of the tangent spaces $T_{\bm{M}_k}\mathcal{M}$ for $k=1,2,\cdots,s$. In the context of individual Riemannian manifolds $\mathcal{M}$, the tangent space $T_{\bm{M}_k}\mathcal{M}$ can be decomposed into two complementary vector space in the sense of the Riemannian metric $g_{\bm{M}_k}$, given as 
$
 T_{\bm{M}_k}\mathcal{M} = \mathcal{V}_{\bm{M}_k}\mathcal{M}\oplus\mathcal{H}_{\bm{M}_k}\mathcal{M},
$
 where vectors on the horizontal space $\mathcal{H}_{\bm{M}_k}\mathcal{M}$ are orthogonal to the equivalence class (\ref{equ2}).
For $k = 1,2,\cdots,s$, we individually the search direction $\bm{\eta}_{\bm{M}_k}$ on the horizontal space $\mathcal{H}_{\bm{M}_k}\mathcal{M}$ and individually retract the directions to the manifold $\mathcal{M}$ via the retraction mapping $\mathcal{R}_{\bm{M}_k}$.
 To sum up, the schematic viewpoint of a Riemannian optimization algorithm developed on the product manifold via element-wise extension is illustrated in Fig. \ref{fig:deri}.
\begin{figure}[tb]
	\centering
	\includegraphics[width=0.6\columnwidth]{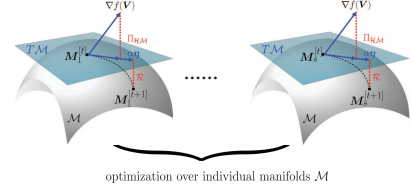}
	\caption{Graphical representation of the concept of Riemannian optimization over individual Riemannian manifolds.}
	\label{fig:deri}
\end{figure}

\section{Optimization Related Ingredients and Algorithms}\label{4.0}
In this section, we provide the optimization related ingredients on product manifold $\mathcal{M}^s$ via the elementwise extension of the optimization related ingredients on the individual Riemannian manifold $\mathcal{M}$. The optimization problem is reformulated via matrix factorization $\bm{M}_k = \bm{w}_k\bm{w}_k^{\mathsf{H}}$, given as
\begin{align}\label{manifold}
\mathop{\textrm{minimize }}_{\bm{v}=\{\bm{w}_k\}_{k=1}^s} f(\bm{v}):=\Big\|\sum\nolimits_{k = 1}^{s}\mathcal{J}_k( \bm{w}_k\bm{w}_k^{\mathsf{H}})-\bm{y}\Big\|^2,
\end{align}
where $\bm{w}_k\in\mathbb{C}^{N+K}$, with $k=1,\cdots,s$. Therefore, according to (\ref{J_M}), the estimated information signal $\hat{\bm{x}}_k$ is represented by the first $N$ rows of the estimated $\hat{\bm{w}}_k$. Compared with problem (\ref{non-convex}) of $2s$ vector variables, i.e., $\{\bm{u}_k,\bm{v}_k\}_{k=1}^s$, problem (\ref{manifold}) endows with only $s$ vector variables, i.e., $\{\bm{w}_k\}_{k=1}^s$, by taking the advantages of the symmetric transformation. To simplify the notations, we abuse the notions of $\bm{v}\in\mathcal{M}^s$, $\bm{w}_k\in\mathcal{M}$.
For convenience, we only take the $k$-th variables as an example and
express all the optimization related ingredients with respect to the complex vector $\bm{w}_k\in\mathcal{M}:=\mathbb{C}_*^{N+K}$, where the space $\mathbb{C}_*^{n}$ denotes the complex Euclidean space $\mathbb{C}^{n}$ removing the origin. 
\subsection{Optimization Related Ingredients}\label{section:opt_ing}
In order to develop the Riemannian optimization algorithm, we need to derive the matrix representations of Riemannian gradient and Riemannian Hessian. To achieve this goal, corresponding optimization related ingredients are required, which are introduced in the following.

The Riemannian metric $g_{\bm{w}_k}:T_{\bm{w}_k}\mathcal{M}\times T_{\bm{w}_k}\mathcal{M}\rightarrow\mathbb{R}$ is an inner product between the tangent vectors on the tangent space $T_{\bm{w}_k}\mathcal{M}$ and invariable along the set of equivalence (\ref{equ2}), which is given by \cite{Yatawatta2013A}
\begin{align}\label{metric}
g_{\bm{w}_k}(\bm{\zeta}_{\bm{w}_k},\bm{\eta}_{\bm{w}_k}) = 
\textrm{Tr}(\bm{\zeta}^{\mathsf{H}}_{\bm{w}_k}\bm{\eta
		}_{\bm{w}_k}+\bm{\eta}^{\mathsf{H}}_{\bm{w}_k}\bm{\zeta}_{\bm{w}_k}), 
\end{align}
where $\bm{\zeta}_{\bm{w}_k},\bm{\eta}_{\bm{w}_k}\in T_{\bm{w}_k}\mathcal{M}
$. Based on the Riemannian metric (\ref{metric}), we further introduce the horizontal space, i.e.,  the matrix representation of the tangent space $T_{\tilde{\bm{w}}_k}(\mathcal{M/\sim})$, the horizontal space projection, which is vital to the derivations of Riemannian gradient and Riemannian Hessian.  \begin{proposition}[Horizontal Space]\label{horizontal_p}
	The quotient manifold $\mathcal{M}/\sim$ admits a horizontal space 
	$\mathcal{H}_{\bm{w}_k}\mathcal{M} \stackrel{\Delta}{=} \lbrace \bm{\eta}_{\bm{w}_k}\in
	\mathbb{C}^{N+K}: \bm{\eta}_{\bm{w}_k}^{\mathsf{H}}\bm{w}_k= \bm{w}_k^{\mathsf{H}}\bm{\eta}_{\bm{w}_k} \rbrace,$ which is the complementary subspace of $\mathcal{V}_{\bm{w}_k}\mathcal{M}$ with respect to the Riemannian metric (\ref{metric}), providing the matrix representation of the abstract tangent space $T_{\tilde{\bm{w}}_k}(\mathcal{M/\sim})$.
	
\end{proposition}
\textit{Proof:}
The proof is mainly based on \cite[Section 3.5]{absil2009optimization}. 
$\hfill\blacksquare$

\begin{proposition}[Horizontal Space Projection]\label{horizontal_pro} The operator $\Pi_{H_{\bm{w}_k}\mathcal{M}}:T_{\bm{w}_k}\mathcal{M}\rightarrow \mathcal{H}_{\bm{w}_k}\mathcal{M}$ that projects vectors on the tangent space onto the horizontal space is called horizontal space projection. It is given by
	$\label{proH}
	\Pi_{\mathcal{H}_{\bm{w}_k}\mathcal{M}}(\bm{\eta}_{\bm{w}_k}) = \bm{\eta}_{\bm{w}_k}-a\bm{w}_k
	$,
	where $a\in\mathbb{C}$ is derived as 
$
	a = \frac{\bm{w}_k^{\mathsf{H}}\bm{\eta}_{\bm{w}_k}-\bm{\eta}_{\bm{w}_k}^{\mathsf{H}}\bm{w}_k}{2\bm{w}_k^{\mathsf{H}}\bm{w}_k}.$
\end{proposition}
\textit{Proof:}
The proof is mainly based on \cite[Section 3.5]{absil2009optimization}.$\hfill\blacksquare$
\subsubsection{Riemannian Gradient}
Let $\tilde{f}:=f\circ\pi$ be the projection of the function $f$ on the quotient space $\mathcal{M}/\sim$.
Consider a point $\bm{w}_k\in\mathcal{M}$ and corresponding point $\tilde{\bm{w}}_k\in\mathcal{M}/\sim$, the matrix representation (horizontal lift) of $\grad_{\tilde{\bm{w}}_k}\tilde{f}\!\in\! T_{\tilde{\bm{w}}_k}(\mathcal{M}/\sim)$， i.e., $\grad_{\bm{w}_k}f\in\mathcal{H}_{\bm{w}}\mathcal{M}$, is required to develop second-order algorithm on the computational space $\mathcal{M}$. Specifically, the Riemannian gradient $\grad_{\bm{w}_k}f$ can be induced from the Euclidean gradient of $f(\bm{v})$ with respect to $\bm{w}_k$. The relationship between them is given by \cite[Section 3.6]{absil2009optimization}
\begin{align}\label{rmgrad}
g_{\bm{w}_k}(\bm{\eta}_{\bm{w}_k},\grad_{\bm{w}_k}f) = \nabla_{\bm{w}_k}f(\bm{v})[\bm{\eta}_{\bm{w}_k}],
\end{align}
where the directional vector is $\bm{\eta}_{\bm{w}_k}\in\mathcal{H}_{\bm{w}_k}\mathcal{M}$ and 
$
\nabla_{\bm{w}_k}f(\bm{v})[\bm{\eta}_{\bm{w}_k}] = \mathop{\textrm{lim}}_{t\rightarrow 0}\frac{f(\bm{v})|_{\bm{w}_k +t\bm{\eta}_k}-f(\bm{v})|_{\bm{w}_k}}{t}.
$
Thus the Riemannian gradient is given by
\begin{align}\label{rgrad}
\grad_{\bm{w}_k}f = \Pi_{\mathcal{H}_{\bm{w}_k}\!\mathcal{M}}(\frac{1}{2}\nabla_{\bm{w}_k}f(\bm{v})),
\end{align}
where $\Pi_{\mathcal{H}_{\bm{w}_k}\!\mathcal{M}}(\cdot)$ is the horizontal space projector operator.
In particular, the Euclidean gradient of $f(\bm{v})$ with respect to $\bm{w}_k$ is represented as
\begin{align}\label{egrad}
\nabla_{\bm{w}_k} f(\bm{v}) = 2\cdot\sum\nolimits_{i = 1}^L(c_{i}\bm{J}_{ki}+c_{i}^{*}\bm{J}_{ki}^{\mathsf{H}})\cdot\bm{w}_k,
\end{align}
where $c_{i} = \sum\nolimits_{k=1}^{s}[\mathcal{J}_k(\bm{w}_k\bm{w}_k^{\mathsf{H}})]_i-y_i$. Details will be demonstrated in Appendix \ref{grad&hess}.
\begin{table*}

	\centering
	\caption{Element-wise Optimization-Related Ingredients For Problem $\mathscr{P}$}\label{table}
	\begin{tabular}{l|l}
	 
		~  &  
		$\mathop{\textrm{minimize }}_{\bm{w}_k\in\mathcal{M}} \|\sum\nolimits_{k = 1}^{s}\mathcal{J}_k( \bm{w}_k\bm{w}_k^{\mathsf{H}})-\bm{y}\|^2$\\
		
		\hline
		
		Computational space $\mathcal{M}$ & $\mathbb{C}_*^{N+K}$ \\	
		Quotient space $\mathcal{M}/\sim$&  $\mathbb{C}_*^{N+K}/\mathrm{SU}(1)$ \\
		Riemannian metric $g_{\bm{w}_k}$ & $g_{\bm{w}_k}({\bm{\zeta_w}}_k,{\bm{\eta_w}}_k) = 
		\textrm{Tr}({\bm{\zeta_w}^{\mathsf{H}}}_k{\bm{\eta
				_w}}_k+{\bm{\eta_w}^{\mathsf{H}}}_k{\bm{\zeta_w}}_k)$ (\ref{metric})
		\\
		Horizontal space $\mathcal{H}_{\bm{w}_k}\mathcal{M}$ & $\bm{\eta}_{\bm{w}_k}\in
		\mathbb{C}^{N+K}: \bm{\eta}_{\bm{w}_k}^{\mathsf{H}}\bm{w}_k= \bm{w}_k^{\mathsf{H}}\bm{\eta}_{\bm{w}_k} $\\
		Horizontal space projection& $\Pi_{\mathcal{H}_{\bm{w}_k}\mathcal{M}}(\bm{\eta}_{\bm{w}_k}) = \bm{\eta}_{\bm{w}_k}-a\bm{w}_k$, $a\in\mathbb{C}$ \\
		
		Riemannian gradient $\textrm{grad}_{\bm{w}_k}f$ & $\textrm{grad}_{\bm{w}}f
		= \frac{1}{2}\nabla_{\bm{w}_k} f(\bm{v})$ (\ref{gradf})
		\\
		Riemannian Hessian $\textrm{Hess}_{\bm{w}_k}f[\bm{\eta}_{\bm{w}_k}]$
		& $\mathrm{Hess}_{\bm{w}_k}f[\bm{\eta_}{\bm{w}_k}] = \Pi_{\mathcal{H}_{\bm{w}_k}\mathcal{M}}(\frac{1}{2}\nabla^2_{\bm{w}_k}f(\bm{v})[\bm{\eta}_{\bm{w}_k}] )$ (\ref{hess})  \\
		Retraction $\mathcal{R}_{\bm{w}_k}:\mathcal{H}_{\bm{w}_k}\mathcal{M} \rightarrow
		\mathcal{M} $
		&       $\mathcal{R}_{\bm{w}_k}(\bm{\eta}_{\bm{w}_k})=\bm{w}_k+\bm{\eta}_{\bm{w}_k}$ (\ref{retraction})

	\end{tabular}
\end{table*}
\subsubsection{Riemannian Hessian}
In order to perform the second-order algorithm, the matrix representation (horizontal lift) of $\mathrm{Hess}_{\tilde{\bm{w}}_k}\tilde{f}[\tilde{\bm{\eta}}_{\bm{w}_k}]\!\in\! T_{\tilde{\bm{w}}_k}\mathcal{M}$, i.e., $\mathrm{Hess}_{\bm{w}_k}f[\bm{\eta}_{\bm{w}_k}]\!\in\!\mathcal{H}_{\bm{w}_k}\mathcal{M}$, needs to be computed via projecting the directional derivative of the Riemannian gradient $\grad_{\bm{w}_k}f$ to the horizontal space $\mathcal{H}_{\bm{w}_k}\mathcal{M}$. 
Based on Propositions \ref{horizontal_p} and \ref{horizontal_pro}, it yields that
\begin{align}\label{hessian}
\textrm{Hess}_{\bm{w}_k}f[\bm{\eta}_{\bm{w}_k}]=\Pi_{\mathcal{H}_{\bm{w}_k}\mathcal{M}}(\nabla_{\bm{\eta}_{\bm{w}_k}}
\textrm{grad}_{\bm{w}_k}f),
\end{align}
where $\textrm{grad}_{\bm{w}_k}f$ (\ref{rmgrad}) denotes the Riemannian gradient and $\nabla_{\bm{\eta}_{\bm{w}_k}} \textrm{grad}_{\bm{w}_k}f$
is the \textit{Riemannian connection}.
In particular, Riemannian connection $\nabla_{\bm{\eta}_{\bm{w}_k}} \bm{\xi}_{\bm{w}_k}$ characterizes the directional derivative of the Riemannian gradient 
and satisfies two properties (i.e., symmetry and invariance of the Riemannian metric) \cite[Section 5.3]{absil2009optimization}. Under the structure of the manifold $\mathcal{M}$, the Riemannian connection is given as
\begin{align}\label{rhess}
\nabla_{\bm{\eta}_{\bm{w}_k}} \bm{\xi}_{\bm{w}_k} = \Pi_{\mathcal{H}_{\bm{w}_k}\mathcal{M}}(D{\bm{\xi}_{\bm{w}_k}} [\bm{\eta}_{\bm{w}_k}]),
\end{align}
where $D{\bm{\xi}_{\bm{w}_k}} [\bm{\eta}_{\bm{w}_k}]$ is the Euclidean directional derivative of $\bm{\xi}_{\bm{w}_k}$ in the direction of $\bm{\eta}_{\bm{w}_k}$.

To derive the Riemannian Hessian (\ref{hessian}), we first compute the directional derivative of Euclidean gradient $\nabla_{\bm{w}_k} f(\bm{v})$ (\ref{egrad}) in the direction of $\bm{\eta}_{\bm{w}_k}\in\mathcal{H}_{\bm{w}_k}\mathcal{M}$, given by
\begin{align}\label{ehess}
\nabla^2_{\bm{w}_k}f(\bm{v})[\bm{\eta}_{\bm{w}_k}] =2\sum\nolimits_{i = 1}^L(b_{i}\bm{J}_{ki}+b_{i}^{*}\bm{J}_{ki}^{\mathsf{H}})\cdot\bm{w}_k+\notag\\
(c_{i}\bm{J}_{ki}+c_{i}^{*}\bm{J}_{ki}^{\mathsf{H}})\cdot\bm{\eta}_{\bm{w}_k},
\end{align}
where
$
b_{i} = \sum\nolimits_{k=1}^s\langle\bm{J}_{ki},\bm{\eta}_{\bm{w}_k}\bm{w}_k^{\mathsf{H}}+\bm{w}_k\bm{\eta}_{\bm{w}_k}^{\mathsf{H}}\rangle.$
According to the formulations (\ref{rgrad}), (\ref{rhess}) and (\ref{ehess}), the Riemannian Hessian is given as
\begin{align}\label{hess}
\mathrm{Hess}_{\bm{w}_k}f[\bm{\eta_}{\bm{w}_k}] = \Pi_{\mathcal{H}_{\bm{w}_k}\mathcal{M}}(\frac{1}{2}\nabla^2_{\bm{w}_k}f(\bm{v})[\bm{\eta}_{\bm{w}_k}] ).
\end{align}
Details will be illustrated in Appendix \ref{grad&hess}.
 To sum up, the element-wise optimization-related ingredients with respect to the manifold $\mathcal{M}$ for the problem (\ref{manifold}) are provided in Table \ref{table}.
 \begin{algorithm}[t]\label{spec_in}
 \footnotesize	
 	\caption{Riemannian gradient descent with spectral initialization}
 	\begin{algorithmic}[1]
 		\renewcommand{\algorithmicrequire}{\textbf{Given:}}
 		\renewcommand{\algorithmicensure}{\textbf{Output:}}
 		
 		\REQUIRE  Riemannian manifold $\mathcal{M}^s$ with optimization-related ingredients, objective function $f$,  $\{\bm{c}_{ij}\}_{1\leq i \leq s, 1\leq j\leq m}$, $\{\bm{b}_j\}_{1\leq j \leq m},$ $\{y_j\}_{1\leq j\leq m}
 		$ and the stepsize $\alpha$.
 		\ENSURE
 		$\bm{v}=\{\bm{w}_k\}_{k=1}^s$
 		\STATE \textbf{Spectral Initialization:}
 		\renewcommand{\algorithmicdo}
 		{\textbf{do in parallel}}
 		\FORALL{$i = 1,\cdots,s$}       
 		\STATE  Let $\sigma_1(\bm{N}_i),~\check{\bm{h}}_i^0$ and $\check{\bm{x}}_i^0$ be the leading singular value, left singular
 		vector and right singular vector of matrix $\bm{N}_i:=\sum_{j=1}^m y_j\bm{b}_j\bm{c}_{ij}^{\mathsf{H}},$
 		respectively.
 		\STATE Set $\bm{w}_i^{[0]} = \left[\begin{matrix}
 		\bm{x}_i^{0}\\\bm{h}_i^{0}
 		\end{matrix}\right]$ where $\bm{x}_i^0 = \sqrt{\sigma_1(\bm{N}_i)} \check{\bm{x}}_i^0$ and $\bm{h}_i^0 = \sqrt{\sigma_1(\bm{N}_i)} \check{\bm{h}}_i^0$.

 		\ENDFOR   
 		\renewcommand{\algorithmicdo}
 		
 		\FORALL{$t = 1,\cdots, T$}
 		\renewcommand{\algorithmicdo}{\textbf{do in parallel}}
 		\FORALL{$i = 1,\cdots,s$}
 		\STATE
 		$\bm{\eta} = -\frac{1}{g_{\bm{w}_k^{[t]}}(\bm{w}_k^{[t]},\bm{w}_k^{[t]})}\grad_{\bm{w}_k^{[t]}}f$
 		\STATE  Update $\bm{w}_{k}^{[t+1]} = \mathcal{R}_{\bm{w}_k^{[t]}}(\alpha_t\bm{\eta})$
 		
 		\ENDFOR
 		\ENDFOR 
 		
 	\end{algorithmic}
 \end{algorithm}
 
 \subsection{Riemannian Optimization Algorithms}
 Based on the optimization related ingredients mentioned in Section \ref{section:opt_ing}, we develop the Riemannian gradient algorithm and Riemannian trust-region algorithm, respectively.
\subsubsection{Riemannian Gradient descent Algorithm}
In the Riemannian gradient descent algorithm, i.e., Algorithm \ref{spec_in}, the search direction is given by $\bm{\eta} = -\grad_{\bm{w}_k^{[t]}}f/{g_{\bm{w}_k^{[t]}}(\bm{w}_k^{[t]},\bm{w}_k^{[t]})}$, where $g_{\bm{w}_k^{[t]}}$ is the Riemannian metric \eqref{metric} and $\grad_{\bm{w}_k^{[t]}}f\in\mathcal{H}_{\bm{w}_k}\mathcal{M}$ is the Riemannian gradient \eqref{rgrad}. Therefore, the sequence of the iterates is given by
 $\bm{w}_{k}^{[t+1]} = \mathcal{R}_{\bm{w}_k^{[t]}}(\alpha_t\bm{\eta})$, where the stepsize $\alpha_t>0$ and
 \begin{align}\label{retraction}
\mathcal{R}_{\bm{w}_k}(\bm{\xi})= \bm{w}_k+\bm{\xi},
 \end{align}
 with $\bm{\xi             }\in\mathcal{H}_{\bm{w}_k}\mathcal{M}$ \cite{Yatawatta2013A}. Here, the retraction map $\mathcal{R}_{\bm{w}_k}:\mathcal{H}_{\bm{w}_k}\mathcal{M} \rightarrow
 \mathcal{M} $ is an approximation of the exponential map that characterizes the motion of ``moving along geodesics on the Riemannian manifold".  More details on computing the retraction are available in \cite[Section 4.1.2]{absil2009optimization}. The statistical analysis of the Riemannian gradient descent algorithm will be provided in the sequel, which demonstrates the linear rate of the proposed algorithm for converging to the ground truth signals.
 \begin{algorithm}[ht]\label{algo1}
 	\caption{Riemannian trust-region algorithm}
 	\begin{algorithmic}[1]
 		\renewcommand{\algorithmicrequire}{\textbf{Given:}}
 		\renewcommand{\algorithmicensure}{\textbf{Output:}}
 		
 		\REQUIRE Riemannian manifold $\mathcal{M}^s$ with Riemannian
 		metric $g_{\bm{v}}$, retraction mapping $\mathcal{R}_{\bm{v}}=\{\mathcal{R}_{\bm{w}_k}\}_{k=1}^s$, objective function $f$ and the stepsize $\alpha$.
 		\ENSURE $\bm{v}=\{\bm{w}_k\}_{k=1}^s$
 		\\ \STATE \textbf{Initialize:} initial point $\bm{v}^{[0]}=\{\bm{w}_k^{[0]}\}_{k=1}^s$, $t=0$   
 		\WHILE{not converged}
 		\renewcommand{\algorithmicdo}{\textbf{do }}
 		\FORALL{$k = 1,\cdots,s$}
 		\STATE  Compute a descent direction $\bm{\eta}$ via implementing trust-region method
 		\STATE  Update $\bm{w}_{k}^{[t+1]} = \mathcal{R}_{\bm{w}_k^{[t]}}(\alpha\bm{\eta})$
 		\STATE  $t=t+1$.

 		\ENDFOR
 		\ENDWHILE	
 	\end{algorithmic}
 \end{algorithm}
\subsubsection{Riemannian Trust-region Algorithm}\label{section:RTR}
 We first consider the setting that searching the direction $\bm{\eta}_{\bm{w}_k}$ on the horizontal space $\mathcal{H}_{\bm{w}_k}\mathcal{M}$, which paves the way to search the direction on the horizontal space $\mathcal{H}_{\bm{V}}\mathcal{M}^s$. At each iteration, let $\bm{w}_k\in\mathcal{M}$, we solve the trust-region sub-problem \cite{absil2009optimization}:
\begin{eqnarray}\label{problem:sub trust-region}
\mathop{\rm{minimize}}_{{\bm{\eta}}_{\bm{w}_k} }&&
{m({\bm{\eta}}_{\bm{w}_k})}\nonumber\\
{\rm{subject~to}}&& g_{\bm{w}_k}({\bm{\eta}}_{\bm{w}_k},{\bm{\eta}}_{\bm{w}_k})\leq
\delta^2,
\end{eqnarray}
where ${\bm{\eta}}_{\bm{w}_k}\in \mathcal{H}_{\bm{w}_k} \mathcal{M}$, $\delta$ denotes the trust-region radius and the cost function
is given by
\begin{align}
m({\bm{\eta}}_{\bm{w}_k}) = g_{\bm{w}_k}({\bm{\eta}}_{\bm{w}_k},\textrm{grad}_{\bm{w}_k}f)+
\frac{1}{2} g_{\bm{w}_k}({\bm{\eta}}_{\bm{w}_k},\textrm{Hess}_{\bm{w}_k}f\left[\bm{\eta}_{\bm{w}_k}\right]),
\end{align}
with $\textrm{grad}_{\bm{w}_k}f$ and $\textrm{Hess}_{\bm{w}_k}f
\left[\bm{\eta}_{\bm{w}_k}\right]$ as the matrix representations of the Riemannian gradient and Riemannian Hessian in the quotient space, respectively. Problem \eqref{problem:sub trust-region} is solved by truncated conjugate gradient method \cite[Section 7.3]{absil2009optimization}. The Riemannian trust-region algorithm is presented in Algorithm \ref{algo1}.

 Let $\bm{w}_k^{[t]}$ denotes the $t$-th iterate. We introduce a quotient to determine whether updating the iterate $\bm{w}_k$ and how to select the trust-region radius implemented in the next iteration. This quotient is given by 
$
\rho_t = (f(\bm{v})\big|_{\bm{w}_k = \bm{w}_k^{[t]}}-f(\bm{v})\big|_{\bm{w}_k = \mathcal{R}_{\bm{w}_k}(\bm{w}_k^{[t]})})/(m({\bm{0}}_{\bm{w}_k^{[t]}})-m({\bm{\eta}}_{\bm{w}_k^{[t]}}))
$ \cite{absil2009optimization}.
The detailed strategy is introduced in the following: reduce the trust-region radius and keep the iterate unchanged, if the quotient $\rho_t$ is extremely small. Expand the trust-region radius and maintain the iterate when the quotient $\rho_t\gg 1$. We update the iterate if and only if the quotient $\rho_t$ is in proper range. The new iterate is given by 
$\bm{w}_{k}^{[t+1]} = \mathcal{R}_{\bm{w}_k^{[t]}}(\bm{\eta}_{\bm{w}_k})
$ \cite{Yatawatta2013A}.

\subsection{Computational Complexity Analysis}
The computational complexity of Algorithm \ref{spec_in} and Algorithm \ref{algo1} depends on the iteration complexity (i.e., number of iterations) and iteration cost of each algorithm. In this subsection, we briefly demonstrate the computational cost in each iteration of the algorithms, which is mainly depends on the computational cost for computing the ingredients listed in Table \ref{table}. More detailed and precise analysis of the iteration complexity of Algorithm \ref{spec_in} and Algorithm \ref{algo1} are left for the future work. 

Note that the linear measurement matrix $\bm{J}_{ki}$ is block and sparse, endowed with computational savings. Define $n = N+K$ and $d = NK$, then the computational cost of these components with respect to manifold $\mathcal{M}$ are showed below.
\begin{enumerate}
	\item Objective function $f(\bm{v})$ (\ref{manifold}): The dominant computational cost comes from computing the terms $\mathcal{J}_k( \bm{w}_k\bm{w}_k^{\mathsf{H}})$, each of which requires a numerical cost of $O(dL)$. Other matrix operations including computing the Euclidean norm of vector of size $L$ involve the computational cost of $O(L)$. Thus the total computational cost of $f(\bm{v})$ is $O(sdL)$.
	\item Riemannian metric $g_{\bm{w}_k}$ (\ref{metric}): The computational cost mainly comes from computing terms $\bm{\zeta}^{\mathsf{H}}_{\bm{w}_k}\bm{\eta}_{\bm{w}_k}$ and $\bm{\eta}^{\mathsf{H}}_{\bm{w}_k}\bm{\zeta}_{\bm{w}_k}$. Each of them requires a numerical cost of $O(n)$. 
	\item Projection on the horizontal space $\mathcal{H}_{\bm{w}_k}\mathcal{M}$ via $\Pi_{\mathcal{H}_{\bm{w}_k}\mathcal{M}}$: As computing the complex value $a$, it involve the multiplications between row vectors of size $n$  and column vectors of size $n$, which costs $O(n)$. Other simple operations (e.g., subtraction and constant multiplication) involve cost of $O(n)$.
	\item Retraction $\mathcal{R}_{\bm{w}_k}$ (\ref{retraction}): The computational cost is $O(n)$.
	\item Riemannian gradient $\textrm{grad}_{\bm{w}_k}f$ (\ref{rgrad}): The dominant computational cost comes from computing multiplications such as $\sum\nolimits_{k=1}^{s}[\mathcal{J}_k(\bm{ww}^{\mathsf{H}})]_i$ and $(c_{i}\bm{J}_{ki}+c_{i}^{*}\bm{J}_{ki}^{\mathsf{H}})\cdot\bm{w}_k$ with the numerical cost of $O(sd)$ and $O(d)$ respectively. Other operations handling with vectors of size $n$ involve the cost of $O(n)$. Therefore, the total computational cost of computing the Riemannian gradient (\ref{rgrad}) is $O(sdL)$.
	\item Riemannian Hessian $\textrm{Hess}_{\bm{w}_k}f[\bm{\eta}_{\bm{w}_k}]$ (\ref{rhess}):
	\begin{itemize}
		\item The directional derivative of Euclidean gradient $\nabla^2_{\bm{w}_k}f(\bm{v})[\bm{\eta}_{\bm{w}_k}]$ (\ref{ehess}): The computational cost for computing this operator is $O(sdL)$, similar as the one of computing $\nabla_{\bm{w}_k} f(\bm{v})$ (\ref{egrad}). 
		\item Projection term: According to the above analysis, the computational cost is $O(n)$.
	\end{itemize}
	All the geometry related operations (e.g., projection and retraction) are of linear computational complexity in $n$, which is computationally efficient. The operations related to the objective problem as well endow with modest computational complexity. By exploiting the admirable geometric structure of symmetric rank-one matrices, the complexity of computing Riemannian Hessian is almost the same as the one of computing Riemannian gradient, which makes second-order Riemannian algorithm yield no extra computational cost compared with first-order algorithms. The numerical results depicted in the next section will demonstrate the computational efficiency of the algorithm.
\end{enumerate}

\section{Main Results}\label{section:MT}
In this section, we provide the statistical guarantees of Riemannian gradient descent, i.e., Algorithm \ref{spec_in}, for the blind demixing problem. 

Without loss of generality, we assume the ground truth $\|\bm{x}_k^\natural\|_2 = \|\bm{h}_k^\natural\|_2$ for $k=  1,\cdots ,s$ and define the condition number $\kappa= \frac{\max_k\|\bm{x}_k^\natural\bm{h}_k^{\natural\mathsf{H}}\|_F}{\min_k \|\bm{x}_k^\natural\bm{h}_k^{\natural\mathsf{H}}\|_F}$ with $\max_k\|\bm{x}_k^\natural\bm{h}_k^{\natural\mathsf{H}}\|_F=1$. Recall the definition of $\bm{w}_k=[\begin{matrix}
\bm{x}_k^{\mathsf{H}}~\bm{h}_k^{\mathsf{H}}
\end{matrix}]^{\mathsf{H}}\in\mathbb{ C}^{N+K}$ and we define the notion $
\bm{v} = [
\bm{w}_1^{\mathsf{H}}\cdots\bm{w}_s^{\mathsf{H}}
]^{\mathsf{H}}\in \mathbb{ C}^{s(N+K)}
$. In practical scenario, the reference symbol for the signal from each user can be exploited to eliminate the ambiguities for blind demixing problem.
In this paper, considering the ambiguities of the estimated signals, we define the discrepancy between the estimate $\bm{v}$ and the ground truth $\bm{v}^{\natural}$ as the distance function, given as
$
\mbox{dist}(\bm{v},\bm{v}^{\natural})=\left(\sum_{i=1}^s\mbox{dist}^2(\bm{v}_i,\bm{v}_i^\natural)\right)^{1/2},
$
where $\mbox{dist}^2(\bm{v}_i,\bm{v}_i^\natural) = \min\limits_{\psi_i\in\mathbb{ C}}({{\|\frac{1}{\overline{\psi_i}}\bm{h}_i-\bm{h}_i^{\natural} \|_2^2+\|\psi_i \bm{x}_i - \bm{x}_i^{\natural}\|_2^2 }})/{d_i}$ for $i = 1,\cdots,s$. Here, $d_i = \|\bm{h}_i^\natural\|^2+\|\bm{x}_i^\natural\|^2$ and each $\psi_i$ is the alignment parameter. In addition, let the incoherence parameter $\mu$ be the smallest number such that
$
\max_{1\leq k \leq s, 1\leq j\leq m}\frac{|\bm{b}^{\mathsf{H}}_j\bm{h}_k^\natural|}{\|\bm{h}_k^\natural\|_2}\leq \frac{\mu}{\sqrt{m}}.
$ The main theorem is presented in the following.
\begin{theorem}\label{mainT}
	Suppose the rows of the encoding matrices, i.e., $\bm{c}_{ij}$'s, follow the i.i.d. complex Gaussian distribution, i.e., $\bm{c}_{ij}\sim \mathcal{N}(0,\frac{1}{2}\bm{I}_N)+i\mathcal{N}(0,\frac{1}{2}\bm{I}_N)$ and the step size obeys $\alpha_t>0$ and $\alpha_t\equiv\alpha\asymp s^{-1}$, then the iterates (including the spectral initialization point) in Algorithm \ref{spec_in} satisfy
	$\mathrm{dist}(\bm{v}^t,\bm{v}^\natural)\leq C_1(1-\frac{\alpha}{16\kappa})^t \frac{1}{\log^2L}
	$
		for all $t\geq 0$ and some constant $C_1>0$, with probability at least $1-c_1L^{-\gamma}-c_1Le^{-c_2K}$ if the number of measurements
	$L\geq C\mu^2s^2\kappa^4\max{\{K,N\}}\log^8 L$ for some constants $\gamma, c_1,c_2>0$ and sufficiently large constant $C>0$.
\end{theorem}
\begin{proof}
	Please refer to Appendix \ref{proof_mainT} for details.   
\end{proof}
 Theorem \ref{mainT} demonstrates that number of measurements $\mathcal{O}(s^2\kappa^4\max{\{K,N\}}\log^8 L)$ are sufficient for the Riemannian gradient descent algorithm (with spectral initialization), i.e., Algorithm \ref{spec_in}, to linearly converge to the ground truth signals.

\section{Simulation Results}\label{4}
In this section, we simulate our proposed Riemannian optimization
algorithm for the blind demixing problem in the settings of Hadamard-type encoding matrices and Gaussian encoding matrices to demonstrate the algorithmic advantages and good performance. In the noiseless scenario (i.e., $\bm{e}=0$), we will study the number of measurements necessary for exact recovery in order to illustrate the performance of the proposed algorithm. The convergence rate of different algorithms will be also compared. In the noisy scenario, we study the average relative construction error of different algorithms. The robustness of the proposed algorithm for noisy data is simultaneously demonstrated.
\subsection{Simulation Settings and Performance Metric} 
The simulation settings are given as follows:
\begin{itemize}
\item Ground truth rank-one matrices $\{\hat{\bm{X}}_k\}_{k = 1}^s$: In the OFDM system, the elements of symbol $\bm{q}_k\in \mathbb{R}^N$ are chosen randomly from the integer set $\{0, 1,\cdots, 14,15\}$. The complex vector $\bm{x}_k$ is thus generated according to (\ref{x}), where $\bm{s}\in
\mathbb{C}^N$ is the 16-QAM symbol constellation corresponding to the symbol $\bm{q}$. In the other scenario, entries of the standard complex Gaussian vector $\bm{x}_k$ are drawn i.i.d from the standard normal distribution. With standard complex Gaussian vectors $\bm{h}_k\in\mathbb{C}^{K}$ and the complex vector $\bm{x}_k\in\mathbb{C}^N$, the matrices $\{\hat{\bm{X}}_k\}_{k=1}^s$ are generated as $\{\hat{\bm{X}}_k\}_{k = 1}^s=\{\bm{x}_k\bar{\bm{h}}_k^{\mathsf{H}}\}_{k=1}^s$ \cite{ling2017regularized}. 

\item Measurement matrices $\{\{\bm{J}_{ik}\}_{i=1}^L\}_{k=1}^s$: We generate the normalized discrete Fourier transform (DFT) matrix $\bm{F}\in\mathbb{C}^{L\times L}$ and the Hadamard-type matrix $\bm{C}_k\in\mathbb{C}^{L\times N}$ according to (\ref{C}) for $k = 1,\cdots,s$ and to construct the measurement matrices according to (\ref{linear_op}) and (\ref{B}).

\item Performance metric: The relative construction error is adopted to evaluate the performance of the algorithms, given as \cite{ling2017regularized}
\begin{align}\label{re}
\mathrm{err}(\bm{X})=\frac{\sqrt{\sum\nolimits_{k=1}^{s}\|\bm{X}_k-\hat{\bm{X}}_k\|_F^2}}{\sqrt{\sum\nolimits_{k=1}^s\|\hat{\bm{X}}_k\|_F^2}},
\end{align}
where $\{\bm{X}_k\}_{k=1}^s$ are estimated matrices and $\{\hat{\bm{X}}_k\}$ are ground truth matrices.

\end{itemize}

\begin{figure}[htbp]
	\centering                                        
	\subfigure[Gaussian encoding matrices]{                    \begin{minipage}{7cm}\centering                                                          
	 \includegraphics[scale=0.8]{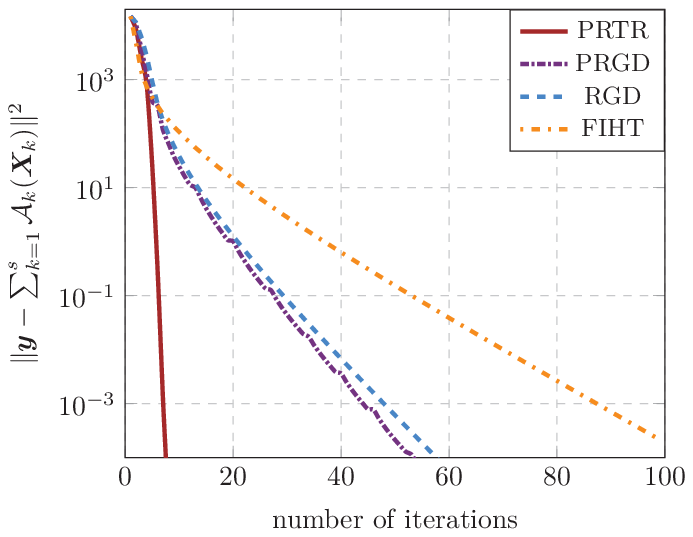}             
	 \end{minipage}\label{fig2_o}}
 \subfigure[Hadamard-type encoding matrices]{                  
 \begin{minipage}{7cm}\centering                                                          \includegraphics[scale=0.8]{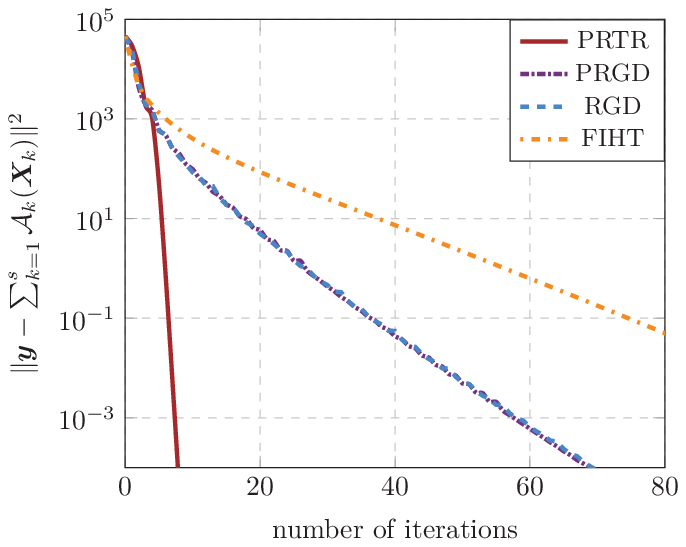}             
 \end{minipage}\label{fig2}}
 \caption{Convergence rate of different algorithms with respect to the number of iterations.}
\end{figure}

\begin{figure}[htbp]
	\centering                                        
	\subfigure[Gaussian encoding matrices]{                    \begin{minipage}{7cm}\centering                                                          
			\includegraphics[scale=0.8]{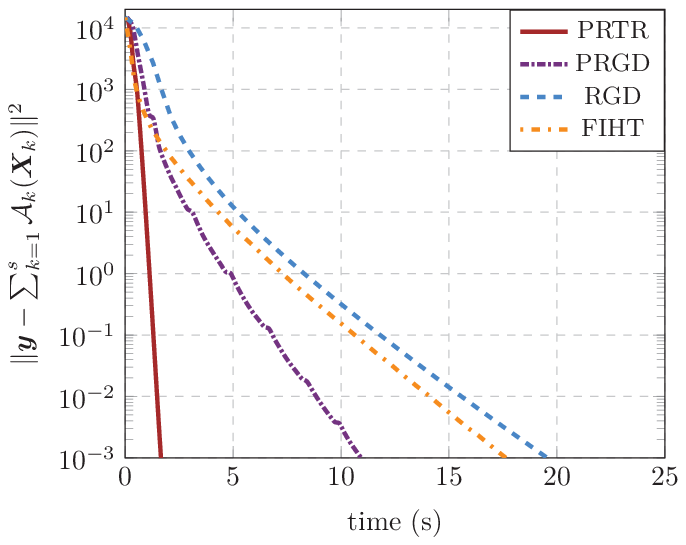}             
	\end{minipage}\label{fig2_1_o}}
\subfigure[Hadamard-type encoding matrices]{                  
		\begin{minipage}{7cm}\centering                                                          \includegraphics[scale=0.8]{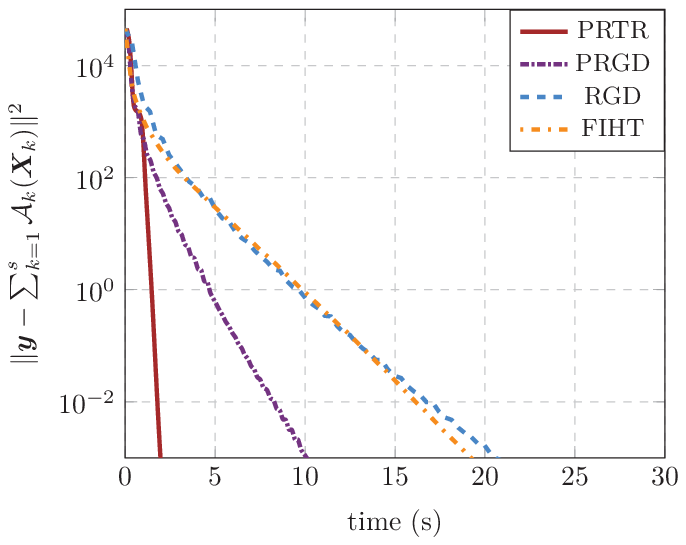}             
	\end{minipage} \label{fig2_1} }
	\caption{Convergence rate of different algorithms with respect to time.}
\end{figure}

\begin{figure}[tb]
	\centering
	\includegraphics[width=0.65\columnwidth]{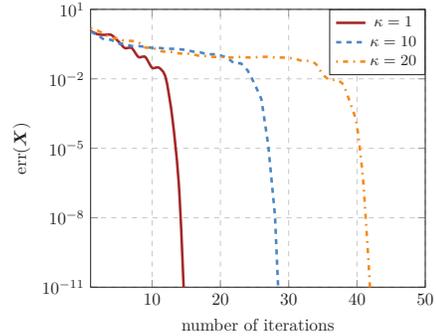}
	\caption{Convergence rate of the proposed trust-region algorithm with respect to different $\kappa$.}
	\label{fig_con}
\end{figure} 
The following five algorithms are compared:
\begin{itemize}
	\item \textbf{Proposed Riemannian trust-region algorithm (PRTR)}: We use the manifold optimization toolbox \emph{Manopt} \cite{manopt} to implement the proposed Riemannian trust-region algorithm (PRTR). The initial trust region radius is $\delta = 2$.
	\item \textbf{Proposed Riemannian gradient descent algorithm (PRGD)}: The proposed Riemannian gradient descent algorithm (PRGD) is implemented via the manifold optimization toolbox \emph{Manopt} \cite{manopt}. 
	\item \textbf{Nuclear norm minimization (NNM)}: The algorithm \cite{ahmed2014blind} is implemented with the toolbox CVX \cite{grant2008cvx} to solve the convex problem (\ref{convex_re}) with the parameter $\varepsilon=10^{-9}$ in noiseless scenario and $\varepsilon=10^{-2}$ in noisy scenario.
	\item \textbf{Regularized gradient descent (RGD)}: This algorithm \cite{ling2017regularized} is implemented to solve the regularized problem (\ref{non-convex}).

	\item \textbf{Fast Iterative Hard Thresholding (FIHT)}: The algorithm \cite{strohmer2017painless} utilizes hard thresholding to solve the rank-constraint problem $\mathscr{P}$ directly.
\end{itemize}

We adopt the initialization strategy in \cite{strohmer2017painless} for all the non-convex optimization algorithms (i.e., PRTR, PRGD, RGD and FIHT).
The PRTR and PRGD algorithm stop when the norm of Riemannian gradient falls below $10^{-8}$ or the number of iterations exceeds $500$. The stopping criteria of RGD and FIHT are based on \cite{ling2017regularized} and \cite{strohmer2017painless}, respectively. 
 \subsection{Convergence Rate}
 Fig. \ref{fig2_o} and Fig. \ref{fig2} illustrate the convergence rate of different non-convex algorithms with respect to the number of iterations in the setting of Gaussian encoding matrices with $N=K=50$, $L=1250$, $s=5$ and in the setting of Hadamard-type encoding matrices with $N=K=50$, $L=1536$, $s=5$, respectively. Under the corresponding settings, Fig. \ref{fig2_1_o} and Fig. \ref{fig2_1} show the convergence rate of different non-convex algorithms  with respect to the time in the settings of Hadamard-type encoding matrices and Gaussian encoding matrices, respectively. From these figures, we can see that in both scenarios, the iteration complexity of the proposed Riemannian gradient descent algorithm is comparable to the regularized gradient descent and has lower time complexity that regularized gradient descent. Moreover, the Riemannian trust-region algorithm, which enjoys superlinear convergence rate, significantly converges faster than the stat-of-the-art non-convex algorithms with respect to both the number of iterations and time. The proposed algorithm thus enjoys low iteration complexity and low time complexity. Next we investigate the impact of the condition number, i.e., $\kappa= \frac{\max\|\hat{X}_k\|_F}{\min \|\hat{X}_k\|_F}$ where $\hat{X}_k$ is the ground truth, on the convergence rate of the proposed Riemannian algorithm. In this simulation, we set $s=2$ and set the first component as $\|\hat{X}_1\|_F = 1$ and the second one as $\|\hat{X}_2\|_F = \kappa$ where $\kappa\in\{1,10,20\}$. Therein, $\kappa = 1$ means both sensors receive the signals with equal power and $\kappa = 10$ means the second sensor has considerably stronger received signals \cite{ling2017regularized}. Fig. \ref{fig_con} demonstrates the relative error (\ref{re}) vs. iterations for the proposed Riemannian trust-region algorithm. It shows that even though large condition number yields slightly slow convergence rate, the proposed Riemannian trust-region algorithm can still precisely recover the original signals in a few iterations. However, the gradient descent algorithm in \cite{ling2017regularized} has less satisfied signal recovery performance when the condition number is large. Therefore, our proposed second-order algorithm is robust to the condition number compared with the first-order algorithm in \cite{ling2017regularized}.

\subsection{Phase Transitions}
In this subsection, we investigate the empirical recovery performance of PRTR without considering the noise and compare the proposed algorithm with other algorithms. In setting of Gaussian encoding matrices, we set $N=K=50$, $L =1000$ with the number of devices $s$ varying from $1$ to $12$. In the setting of Hadamard-type encoding matrices, we set $N=K=16$, $L=1536$ with $s$ varying from $1$ to $45$. For each setting, $10$ independent trails are performed and the recovery is treated as a success if the relative construction error $\mathrm{err}(\bm{X})\leq 10^{-3}$. Fig. \ref{fig1_o} and Fig. \ref{fig1} show the probability of successful recovery for different numbers of devices $s$ in the settings of Hadamard-type encoding matrices and Gaussian encoding matrices, respectively. Based on the phase transitions results in two figures, we can see that the PRTR and PRGD algorithm outperform in terms of guaranteeing exact recovery than other three algorithms. The non-uniqueness of the factorization taken into account in the quotient manifold space plays a vital role to lead this advantage. In particular, the proposed algorithm, i.e., the Riemannian gradient algorithm and the Riemannian trust-region algorithm are more robust to the practical scenario than NNM and FIHT algorithm do.

\begin{figure}[htbp]
	\centering                                        
	\subfigure[Gaussian encoding matrices]{                    \begin{minipage}{7cm}\centering                                                          
			\includegraphics[scale=0.8]{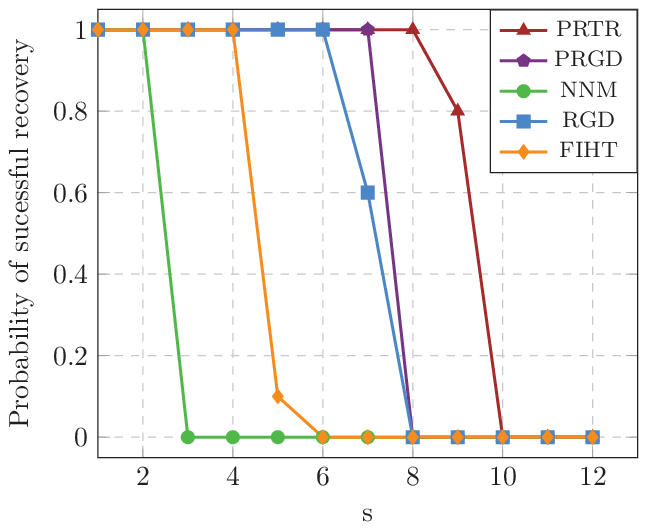}             
	\end{minipage}\label{fig1_o}}	
	\subfigure[Hadamard-type encoding matrices]{                  
		\begin{minipage}{7cm}\centering                                                          \includegraphics[scale=0.8]{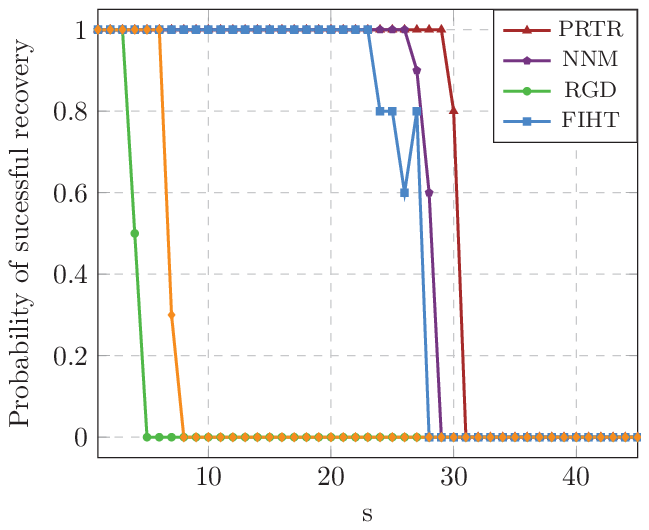}             
	\end{minipage} \label{fig1}} 
	\caption{Probability of successful recovery with different numbers of devices $s$.}
\end{figure}

\begin{figure}[htbp]
	\centering                                        
	\subfigure[Gaussian encoding matrices]{                    \begin{minipage}{7cm}\centering                                                          
			\includegraphics[scale=0.8]{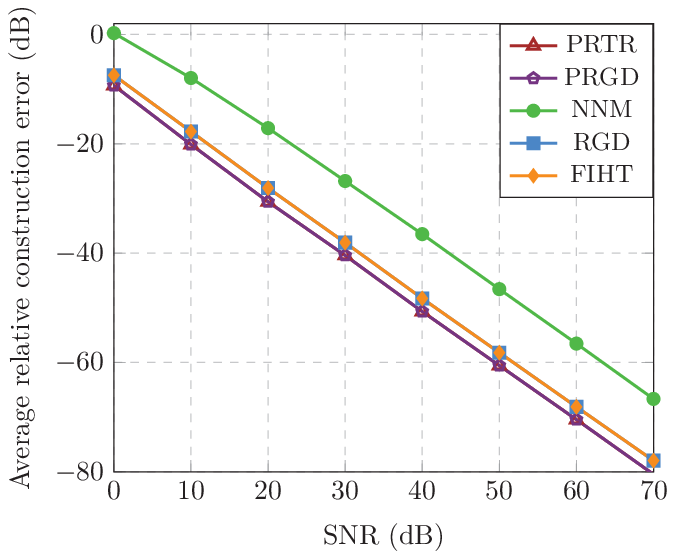}             
	\end{minipage}\label{fig3_o}}	
	\subfigure[Hadamard-type encoding matrices]{                  
		\begin{minipage}{7cm}\centering                                                          \includegraphics[scale=0.8]{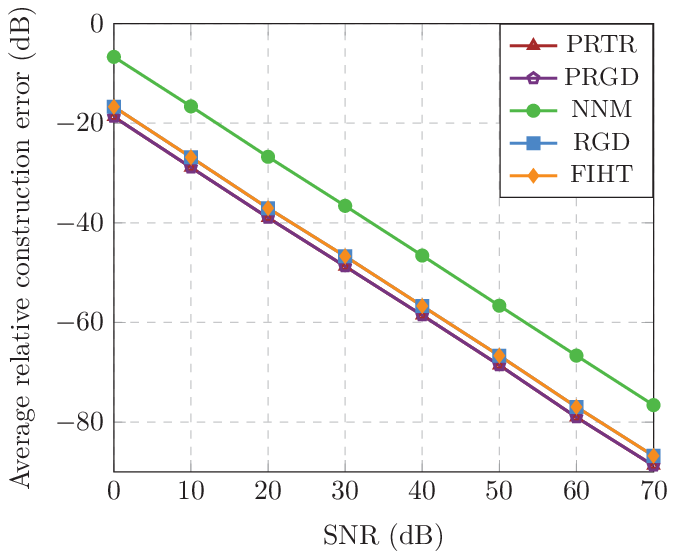}             
	\end{minipage}\label{fig3} } 
	\caption{Average relative construction error vs. SNR (dB).}
\end{figure}

\subsection{Average Relative Construction Error}
We study the average relative construction error of four algorithms and explore the robustness of the proposed Riemannian trust-region algorithm against additive noise by considering the model (\ref{equ}).
We assume the additive noise in the formulation (\ref{equ}) satisfies \cite{strohmer2017painless}
\begin{align}
\bm{e} = \sigma\cdot \|\bm{y}\|\cdot\frac{\bm{\omega}}{\|\bm{\omega}\|},
\end{align}
where $\bm{\omega}\in\mathbb{C}^L$ is a standard complex Gaussian vector. We compare the four algorithms for each level of signal to noise ratio (SNR) $\sigma$ in the setting of Gaussian encoding matrices with $L = 1500$, $N=K=50$, $s=2$ and in the setting of Hadamard-type encoding matrices with $L = 1536$, $N=K=16$ and $s=2$. For each setting, $10$ independent trails are performed and the condition of successful recovery is the same with the one aforementioned.  The average relative construction
error in dB against the signal to noise ratio (SNR) in the settings of Hadamard-type encoding matrices and Gaussian encoding matrices are illustrated in Fig. \ref{fig3_o} and Fig. \ref{fig3}, respectively. It depicts that the relative reconstruction error of the proposed algorithm linearly scales with SNR. We conclude that the PRTR and PRGD algorithm are stable in the presence of noise as the other algorithms. The figure also shows that the proposed algorithm PRTR and PRGD achieve lower average relative construction error than other three algorithms, yielding better performance in the practical scenario. 

The impressive simulation results are in favor of the quotient manifold of the product of complex symmetric rank-one matrices which is established via Hermitian reformulation. By exploiting the geometry of quotient manifold which takes into account the non-uniqueness of the factorization, both first-order and second-order Riemannian optimization algorithm are developed. Specifically, the proposed Riemannian optimization algorithms, i.e., the Riemannian gradient descent algorithm and the Riemannian trust-region algorithm outperform state-of-the-art algorithms in terms of the algorithmic advantages (i.e., fast convergence rate and low iteration cost) and performance (i.e., sample complexity). Moreover, both of them are robust to the noise and the Riemannian trust-region algorithm is also robust to the condition number, i.e., $\kappa= \frac{\max\|\hat{X}_k\|_F}{\min \|\hat{X}_k\|_F}$.
\section{Conclusion and Future Work}\label{5}
In this paper, we presented the blind demixing approach to support low-latency communication without any channel estimation in OFDM system, for which a low-rank modeling approach with respect to rank-one matrices is further developed. To address the unique challenge of multiple asymmetric complex rank-one matrices as well as develop efficient algorithm, we exploited the Riemannian quotient geometry of product of complex symmetric rank-one matrices via reformulating problems on complex asymmetric matrices to problems on Hermitian positive semidefinite (PSD) matrices. Specifically, by exploiting the admirable structure of symmetric rank-one matrices, we developed scaled Riemannian optimization algorithms, i.e., the Riemannian gradient descent algorithm and the Riemannian trust-region algorithm. We proof that the Riemannian gradient descent algorithm linearly converges to ground truth signals with high probability provided sufficient measurement. Simulation results demonstrated that the proposed algorithms are robust to the additive noise and outperforms state-of-the-art algorithms in terms of algorithmic advantages and performance in the practical scenario. Moreover, the Riemannian trust-region algorithm is robust to the condition number.
\appendices

\section{Computing the Riemannian gradient and Riemannian Hessian}\label{grad&hess}
We first reformulate the objective function in (\ref{manifold}) as
\begin{align}
f(\bm{v})& = \sum\nolimits_{i = 1}^{L}(\sum\nolimits_{k=1}^s[\mathcal{J}_k(\bm{w}_k\bm{w}_k^{\mathsf{H}})]_i-y_i)^{*}\notag\\&(\sum\nolimits_{k=1}^s[\mathcal{J}_k(\bm{w}_k\bm{w}_k^{\mathsf{H}})]_i-y_i)=\sum\nolimits_{i=1}^L c_{i}^*\cdot c_{i},
\end{align} 
where $c_{i} = \sum\nolimits_{k=1}^{s}[\mathcal{J}_k(\bm{w}_k\bm{w}_k^{\mathsf{H}})]_i-y_i$.
The partial gradient of $f(\bm{v})$ with respect to the complex vector $\bm{w}_k$ is given as
\begin{align}
\nabla_{\bm{w}_k}f(\bm{v}) =& 2\frac{\partial f(\bm{v})}{\partial\bm{w}_k^{\mathsf{H}}}\notag\\
=&  2\sum\nolimits_{i=1}^L (\sum\nolimits_{k=1}^{s}[\mathcal{J}_k(\bm{w}_k\bm{w}_k^{\mathsf{H}})]_i-y_i)\bm{J}_{ki}\bm{w}_k+\notag\\&(\sum\nolimits_{k=1}^{s}[\mathcal{J}_k(\bm{w}_k\bm{w}_k^{\mathsf{H}})]_i-y_i)^{*}\bm{J}_{ki}^{\mathsf{H}}\bm{w}_k\notag\\
=& 2\cdot\sum\nolimits_{i = 1}^L(c_{i}\bm{J}_{ki}+c_{i}^{*}\bm{J}_{ki}^{\mathsf{H}})\cdot\bm{w}_k.
\end{align}
Furthermore, the Euclidean gradient of $f(\bm{v})$ in the direction $\bm{\eta}_{\bm{w}_k}$ with respect to $\bm{w}_k$ is derived as
\begin{align}
\nabla_{\bm{w}_k}f(\bm{v})[\bm{\eta}_{\bm{w}_k}]  &=\mathop{\textrm{lim}}_{t\rightarrow 0}\frac{f(\bm{v})|_{\bm{w}_k +t\bm{\eta}_k}- f(\bm{v})|_{\bm{w}_k}}{t}\notag\\
& = \trace\Big[\sum\nolimits_{i=1}^L \bm{\eta}_{\bm{w}_k}^{\mathsf{H}}(c_{i}\bm{J}_{ki}+c_{i}^{*}\bm{J}_{ki}^{\mathsf{H}})\bm{w}_k\notag\\&+\bm{w}_k^{\mathsf{H}}(c_{i}\bm{J}_{ki}+c_{i}^{*}\bm{J}_{ki}^{\mathsf{H}})\bm{\eta}_{\bm{w}_k}\Big]\notag\\
& = g_{\bm{w}_k}\big(\bm{\eta}_{\bm{w}_k},\sum\nolimits_{i = 1}^L(c_{i}\bm{J}_{ki}+c_{i}^{*}\bm{J}_{ki}^{\mathsf{H}})\cdot\bm{w}_k\big)\notag\\
& = g_{\bm{w}_k}\big(\bm{\eta}_{\bm{w}_k},\frac{1}{2}\nabla_{\bm{w}_k}f(\bm{v})\big).
\end{align}
Thus, according to (\ref{rmgrad}), there is $\mathrm{grad}_{\bm{w}_k}f = \frac{1}{2} \nabla_{\bm{w}_k}f(\bm{v})$. With the fact that $(\mathrm{grad}_{\bm{w}_k}f)^{\mathsf{H}}\bm{w}_k = \bm{w}_k^{\mathsf{H}}\mathrm{grad}_{\bm{w}_k}f$, we conclude that $\mathrm{grad}_{\bm{w}_k}f$ is already in the horizontal space $\mathcal{V}_{\bm{w}_k}\mathcal{M}$. Therefore, the matrix representation of Riemannian gradient is written as 
\begin{align}\label{gradf}
\mathrm{grad}_{\bm{w}_k}f = \sum\nolimits_{i = 1}^L(c_{i}\bm{J}_{ki}+c_{i}^{*}\bm{J}_{ki}^{\mathsf{H}})\cdot\bm{w}_k.
\end{align}

To compute the Riemannian Hessian, we first derive the directional derivative of Euclidean gradient $\nabla_{\bm{w}_k} f(\bm{v})$ (\ref{egrad}) in the direction of $\bm{\eta}_{\bm{w}_k}\in\mathcal{H}_{\bm{w}_k}\mathcal{M}$, given by
\begin{align}\label{hess_e}
\nabla^2_{\bm{w}_k}f(\bm{v})[\bm{\eta}_{\bm{w}_k}]=& \mathop{\textrm{lim}}_{t\rightarrow 0}\frac{\nabla_{\bm{w}_k} f(\bm{v})|_{\bm{w}_k +t\bm{\eta}_k}-\nabla_{\bm{w}_k} f(\bm{v})|_{\bm{w}_k}}{t}\notag\\
=&\frac{\mathrm{d}}{\mathrm{d}t}\bigg|_{t=0}\nabla_{\bm{w}_k} f(\bm{v})|_{\bm{w}_k+t\bm{\eta}_k}\notag\\
=&2\sum\nolimits_{i = 1}^L(b_{i}\bm{J}_{ki}+b_{i}^{*}\bm{J}_{ki}^{\mathsf{H}})\cdot\bm{w}_k\notag\\&+
(c_{i}\bm{J}_{ki}+c_{i}^{*}\bm{J}_{ki}^{\mathsf{H}})\cdot\bm{\eta}_{\bm{w}_k},
\end{align}
where $b_{i} = \sum\nolimits_{k=1}^s\langle\bm{J}_{ki},\bm{\eta}_{\bm{w}_k}\bm{w}_k^{\mathsf{H}}+\bm{w}_k\bm{\eta}_{\bm{w}_k}^{\mathsf{H}}\rangle$.
Thus, the matrix representation of Riemannian Hessian is derived according to (\ref{hess}).

\section{Proof of Theorem \ref{mainT}}\label{proof_mainT}
According to the definition of the horizontal space in 
Proposition 1, we know that $\nabla_{\bm{w}_k}f(\bm{v})$ is in the horizontal space due to $\nabla_{\bm{w}_k}f(\bm{v})^{\mathsf{H}}\bm{w}_k = \bm{w}_k^{\mathsf{H}}\nabla_{\bm{w}_k}f(\bm{v})$. Based on this fact, the update rule in the Riemannian gradient descent algorithm, i.e., Algorithm \ref{spec_in}, can be reformulated as
\begin{align}\label{update}
\bm{w}_k^{[t+1]}= \bm{w}_k^{[t]}-\frac{\alpha_t}{2\|\bm{w}_k^{[t]}\|_2^2}\nabla_{\bm{w}_k}f(\bm{v})|_{\bm{w}_k^{[t]}},
\end{align}
according to the definition of the Riemannian metric $g_{\bm{w}_k}$ \eqref{metric} and the retraction $\mathcal{R}_{\bm{w}_k}$  \eqref{retraction}. The update rule (\ref{update}) can be further modified as \begin{align}\label{update_re}
\left[~\begin{matrix}
\bm{w}_k^{[t+1]}\\\overline{\bm{w}_k^{[t+1]}}
\end{matrix}~\right] = \left[~\begin{matrix}
\bm{w}_k^{[t]}\\\overline{\bm{w}_k^{[t]}} 
\end{matrix}~\right]-\frac{\alpha_t}{\|\bm{w}_k^{[t]}\|_2^2}\left[~\begin{matrix}
\frac{\partial f}{\partial \bm{w}_k^{\mathsf{H}}}|_{\bm{w}_k^{[t]}}  
\\\overline{\frac{\partial f}{\partial \bm{w}_k^{\mathsf{H}}}}|_{\bm{w}_k^{[t]}}  
\end{matrix}~\right],
\end{align}based on the fact that $\nabla_{\bm{w}_k}f(\bm{v}) = 2\frac{\partial f(\bm{v})}{\partial\bm{w}_k^{\mathsf{H}}}$.

To proof Theorem \ref{mainT}, under the assumption that the rows of the encoding matrices $\bm{c}_{ij}\sim \mathcal{N}(0,\frac{1}{2}\bm{I}_N)+i\mathcal{N}(0,\frac{1}{2}\bm{I}_N)$,  we first characterize the local geometry in the region of incoherence and contraction (RIC) where the objective function
enjoys restricted strong convexity and smoothness near
the ground truth $\bm{v}^{\natural}$,
please refer to Lemma \ref{L1}. The error contraction, i.e., convergence analysis, is further established in Lemma \ref{L2} based on the property of the local geometry. We then exploit the induction arguments to demonstrate that the iterates of Algorithm \ref{spec_in}, including the spectral initialization point, stay within the RIC, please refer to Lemma \ref{L3}.

\begin{definition}[$(\phi,\beta,\gamma,\bm{z}^\natural)-\mathcal{R}$ the region of incoherence and contraction]Define $\bm{z}_i=[\begin{matrix}
	\bm{x}_i^{\mathsf{H}}~\bm{h}_i^{\mathsf{H}}
	\end{matrix}]^{\mathsf{H}}\in\mathbb{ C}^{N+K}$ and $
	\bm{z} = [
	\bm{z}_1^{\mathsf{H}}\cdots\bm{z}_s^{\mathsf{H}}
	]^{\mathsf{H}}\in \mathbb{ C}^{s(N+K)}
	$. If $\bm{z}$ is in the region of incoherence and contraction (RIC), i.e., $\bm{z}\in(\phi,\beta,\gamma,\bm{z}^\natural)-\mathcal{R}$, it holds that
	\begin{subequations}\label{con85}
	\begin{align}
	&\quad\quad\mathrm{dist}(\bm{z}^{t},\bm{z}^\natural)\leq \phi,\label{16a}\\
	&       \max_{1\leq i \leq s,1\leq j \leq m}\left|\bm{c}_{ij}^{\mathsf{H}}\left(\widetilde{\bm{x}}_i^t-\bm{x}_i^\natural\right)\right|\cdot\|\bm{x}_i^\natural\|_2^{-1}\leq {C_3}\beta,\\
	&\max_{1\leq i \leq s,1\leq j \leq m}\left|\bm{b}_{j}^{\mathsf{H}}\widetilde{\bm{h}}_i^t\right|\cdot\|\bm{h}_i^\natural\|_2^{-1}\leq {C_4}\gamma,\label{con85c}
	\end{align}
	for some constants $C_3,C_4>0$ and some sufficiently small constants $\phi, \beta ,\gamma>0$. Here, $\widetilde{\bm{h}}^t_i$ and $\widetilde{\bm{x}} ^t_i
	$ are defined as $\widetilde{\bm{h}}^t_i = \frac{1}{\overline{\psi_i^t}}\bm{h}_i^t$ and $\widetilde{\bm{x}} ^t_i= \psi_i^t\bm{x}_i^t$ for $i = 1,\cdots, s$, where $ \psi_i^t$ is the alignment parameter.
\end{subequations}
\end{definition}

The Riemannian Hessian is denoted as $
{\hess }f(\bm{v}):=\diagg(\{{\hess }_{\bm{w}_i}f\}_{i=1}^s).
$
\begin{lemma}\label{L1}
Suppose there is a sufficiently small constant $\delta>0$. If the number of measurements obeys $m\gg \mu^2s^2\kappa^2\max{\{N,K\}}\log^5m $, then with probability at least $1-\mathcal{O}(m^{-10})$, the Riemannian Hessian ${\hess }f(\bm{v})$ obeys
	\begin{align}
	\bm{u}^{\mathsf{H}}\left[\bm{D}{\hess }f(\bm{v})  + {\hess }f(\bm{v})\bm{D}\right]\bm{u}\geq \frac{1}{4\kappa}\|\bm{u}\|_2^2~~{\text{and}}
	~~\left\|{\hess }f(\bm{v})\right\|\leq 2+s
	\end{align}
	simultaneously for all 
	$
	\bm{u} = [~\begin{matrix}
	\bm{u}_1^{\mathsf{H}}~\cdots~\bm{u}_s^{\mathsf{H}}
	\end{matrix}~]^{\mathsf{H}}~\textrm{with}~\bm{u}_i = [~\begin{matrix}
	(\bm{x}_i-\bm{x}_i^\prime)^{\mathsf{H}}~
	(\bm{h}_i-\bm{h}_i^\prime)^{\mathsf{H}}~
	{(\bm{x}_i-\bm{x}_i^\prime)}^{\top}~
	{(\bm{h}_i-\bm{h}_i^\prime)}^{\top}
	\end{matrix}~]^{\mathsf{H}},
$ and $\bm{D} = \diagg\left(\{\bm{W}_i\}_{i=1}^s\right)$ with $
\bm{W}_i = \diagg\left(\left[
	\overline{\beta}_{i1}\bm{I}_K~
	\overline{\beta}_{i2}\bm{I}_N~
	\overline{\beta}_{i1}\bm{I}_K~
	\overline{\beta}_{i2}\bm{I}_N
	\right]^*     \right).$

	Here $\bm{v}$ is in the region $(\delta,\frac{1}{\sqrt{s}\log^{3/2}m},\frac{\mu}{\sqrt{m}}\log^{2}m,\bm{v}^\natural)-\mathcal{R}$, and one has
	$
	\max \{\|\bm{h}_i-\bm{h}_i^\natural\|_2,\|\bm{h}_i^\prime-\bm{h}_i^\natural\|_2,\|\bm{x}_i-\bm{x}_i^\natural\|_2,
	\|\bm{x}_i^\prime-\bm{x}_i^\natural\|_2\} \leq{\delta}/({\kappa\sqrt{s}}),
	$
	for $i=  1,\cdots,s$ and $\bm{W}_i$'s satisfy that for $\beta_{i1},\beta_{i2}\in\mathbb{ R}$, for $i=  1,\cdots,s$
	$
	\max_{1\leq i \leq s}\max\left\{|\beta_{i1}-\frac{1}{\kappa}|,|\beta_{i2}-\frac{1}{\kappa}|\right\}\leq \frac{\delta}{\kappa\sqrt{s}}.
	$
	Therein, $C_3,C_4\geq 0$ are numerical constants.
\end{lemma}

\begin{lemma}\label{L2}
	Suppose the number of measurements satisfies $m\gg\mu^2s^2\kappa^4\max{\{N,K\}}\log^5m$ and the step size obeys $\alpha_t>0$ and $\alpha_t\equiv\alpha\asymp s^{-1}$. Then with probability at least $1-\mathcal{O}(m^{-10})$,
	\begin{align}
		&\mathrm{dist}(\bm{v}^{t+1},\bm{v}^\natural)\leq (1-\frac{\alpha}{16\kappa})\mathrm{dist}(\bm{v}^t,\bm{v}^\natural),
	\end{align}
	provided that
	$\bm{v}$ is in the region $(\delta,\frac{1}{\sqrt{s}\log^{3/2}m},\frac{\mu}{\sqrt{m}}\log^{2}m,\bm{v}^\natural)-\mathcal{R}$.
\end{lemma}

\begin{lemma}\label{L3}
	The spectral initialization point $\bm{v}^0$ is in the region $(\frac{1}{\log m},\frac{1}{\sqrt{s}\log^{3/2}m},\frac{\mu}{\sqrt{m}}\log^{2}m,\bm{v}^\natural)-\mathcal{R}$ with probability at least $1-\mathcal{O}(m^{-9})$, provided $m\gg \mu^2s^2\kappa^2\max{\{K,N\}}\log^6m$. Suppose $t$-th iteration $\bm{v}^t$ is in the region $(\frac{1}{\log m},\frac{1}{\sqrt{s}\log^{3/2}m},\frac{\mu}{\sqrt{m}}\log^{2}m,\bm{v}^\natural)-\mathcal{R}$ and the number of measurements satisfy $m\gg \mu^2s^2\kappa^2\max{\{K,N\}}\log^8m$. Then with probability at least $1-\mathcal{O}(m^{-9})$,
	the $(t+1)$-th iteration $\bm{v}^{t+1}$ is also in the region $(\frac{1}{\log m},\frac{1}{\sqrt{s}\log^{3/2}m},\frac{\mu}{\sqrt{m}}\log^{2}m,\bm{v}^\natural)-\mathcal{R}$,
	provided that the step size satisfies $\alpha_t>0$ and $\alpha_t\equiv\alpha\asymp s^{-1}$.
\end{lemma}
\begin{remark}
	The proofs of Lemma \ref{L1}, Lemma \ref{L2} and Lemma \ref{L3} are mainly based on the proofs of Lemma 1-Lemma 7 in  \cite{dong}.
\end{remark}

	\bibliographystyle{ieeetr}
	\bibliography{Reference}
	
\end{document}